\documentclass[a4paper,headings=standardclasses]{scrartcl}
\pdfoutput=1
\usepackage[utf8]{inputenc}

\usepackage[T1]{fontenc} 

\usepackage[british]{babel}

\usepackage[osf,p,largesc,nohelv]
  {newpxtext} 
\usepackage{amsmath,amsthm}
\usepackage[varbb]{newpxmath} 
\renewcommand{\partial}{\uppartial} 
\usepackage{mathtools}
\usepackage{tensor}
\usepackage{accents}
\usepackage{enumitem}
\usepackage{geometry}
\usepackage{makecell}
\usepackage{caption}

\usepackage{accsupp}
\newcommand*{\acronym}[1]{\texorpdfstring{%
    \protect\BeginAccSupp{%
      method=pdfstringdef,%
      ActualText=#1
    }%
    \textsc{\textls[40]{\MakeLowercase{#1}}}%
    \protect\EndAccSupp{}}%
  {#1}%
}

\usepackage{csquotes} 

\usepackage{pdflscape} 
\usepackage{afterpage}
\usepackage{booktabs}

\usepackage{microtype} 

\usepackage{authblk}

\usepackage[colorlinks,bookmarksnumbered=true]{hyperref}
\usepackage[nameinlink]{cleveref}

\usepackage[style=ext-authoryear-comp,maxcitenames=2,uniquename=false,
  backend=biber,uniquelist=false,articlein=false]{biblatex}

\DeclareNameAlias{sortname}{family-given}

\DefineBibliographyExtras{british}{\def\finalandcomma{\addcomma}}

\DeclareFieldFormat{eprint:philsci}{%
  \ifhyperref
    {PhilSci-Archive: \href{https://philsci-archive.pitt.edu/#1/}{\nolinkurl{#1}}}
    {PhilSci-Archive: \nolinkurl{#1}}}

\newcommand*{\e}{\mathrm{e}} 
\newcommand*{\D}{\mathrm{d}} 
\newcommand*{\R}{\mathbb R}
\newcommand*{\so}{\mathfrak{so}}
\newcommand*{\gamman}{{^{(n)}\gamma}}
\DeclareMathOperator{\Span}{span}
\DeclareMathOperator{\rk}{rk}

\theoremstyle{plain}
\newtheorem{thm}{Theorem}[section]
\newtheorem{prop}[thm]{Proposition}

\theoremstyle{definition}
\newtheorem{defn}[thm]{Definition}
\newtheorem{rem}[thm]{Remark}
\newtheorem{constr}[thm]{Construction}

\AddToHook{env/prop/begin}{\crefalias{thm}{prop}}
\AddToHook{env/lem/begin}{\crefalias{thm}{lem}}
\AddToHook{env/cor/begin}{\crefalias{thm}{cor}}
\AddToHook{env/defn/begin}{\crefalias{thm}{defn}}
\AddToHook{env/rem/begin}{\crefalias{thm}{rem}}
\AddToHook{env/constr/begin}{\crefalias{thm}{constr}}
\AddToHook{env/exmp/begin}{\crefalias{thm}{exmp}}

\crefname{thm}{theorem}{theorems}
\crefname{prop}{proposition}{propositions}
\crefname{lem}{lemma}{lemmas}
\crefname{cor}{corollary}{corollaries}
\crefname{defn}{definition}{definitions}
\crefname{rem}{remark}{remarks}
\crefname{constr}{construction}{constructions}
\crefname{exmp}{example}{examples}

\makeatletter
\mathchardef\ordinarycolon\mathcode`\:%
\mathcode`\:=\string"8000%
\begingroup \catcode`\:=\active%
  \gdef:{\@ifnextchar{=}
    {\coloneq\@gobble}
    {\ordinarycolon}}%
\endgroup
\AtBeginDocument{\mathcode`\==\string"8000} 
\mathchardef\ordinaryequals\mathcode`\=%
\begingroup \catcode`\==\active%
  \gdef={\@ifnextchar{:}
    {\eqcolon\@gobble}
    {\ordinaryequals}
  }%
\endgroup
\makeatother

\numberwithin{equation}{section}

\title{Weyl-type theorems in Galilei and Carroll geometry}

\author[1,a]{Philip K. Schwartz}
\author[2,b]{James Read}
\author[3,c]{Quentin Vigneron}
\affil[1]{Institute of Theoretical Physics,
  Leibniz University Hannover, \par
  Appelstraße 2, 30167 Hannover, Germany}
\affil[2]{Faculty of Philosophy, University of Oxford, \par
  Oxford \acronym{OX2 6GG}, United Kingdom}
\affil[3]{Centre de Recherche Astrophysique de Lyon, \acronym{ENS} de
  Lyon, Université Claude Bernard Lyon 1, \acronym{CNRS}, \par
  Lyon 69007, France}
\affil[a]{\normalfont\texttt{\href{mailto:philip.schwartz@itp.uni-hannover.de}
    {philip.schwartz@itp.uni-hannover.de}}}
\affil[b]{\normalfont\texttt{\href{mailto:james.read@philosophy.ox.ac.uk}
    {james.read@philosophy.ox.ac.uk}}}
\affil[c]{\normalfont\texttt{\href{mailto:quentin.vigneron@ens-lyon.fr}
    {quentin.vigneron@ens-lyon.fr}}}

\date{}

\begin{document}
\maketitle

\vspace{-3\baselineskip}

\begin{abstract}
  \noindent
  A classic theorem of \textcite{Weyl:1921} states that a \emph{Weyl
    metric}---a natural generalisation of a pseudo-Riemannian
  metric---is uniquely determined by its conformal and projective
  structures (i.e.\ by its conformal structure and its set of
  unparametrised geodesics).  An equivalent formulation of Weyl's
  result is that a torsion-free linear connection compatible with a
  pseudo-Riemannian conformal structure is uniquely determined by its
  projective structure.  We discuss analogous results for suitably
  defined notions of conformal structure for Galilei and Carroll
  geometry, i.e. for spacetime geometries arising as the
  `non-relativistic' and `ultra-relativistic' limits of Lorentzian
  geometry.
\end{abstract}

\section{Introduction}

In \citeyear{Weyl:1921}, Hermann Weyl proved that a given `Weyl
metric'---of which a Lorentzian metric is (in a certain sense to be
made precise below) a special case---is fixed uniquely by its
associated projective and conformal structures (respectively:
unparametrised geodesics of the associated connection, and---in the
case of indefinite signature---lightcone structure).  This result is a
uniqueness theorem in the sense that those projective and conformal
structures fix the Weyl metric uniquely: there are no \emph{other}
Weyl metrics which they can be used to define.  Some fifty years
later, in 1972, \textcite{EPS:1972, EPS:2012} (`\acronym{EPS}') proved
\emph{inter alia} what can understood to be a related existence
theorem: given a projective structure and a conformal structure, and
given further assumptions regarding their `compatibility' (roughly:
the unparametrised timelike geodesics given by the projective
structure all lie inside the timelike lobes of the lightcones given by
the conformal structure), they can be used to construct a Weyl metric.

Since that work, many details have been shored up---see
\textcite{Adlam.Linnemann.Read:2025} and references therein for a
detailed study.  And yet, there also remain technical questions
regarding Weyl's theorem which have yet to be answered decisively.  In
particular, one might ask: is it possible to prove versions of Weyl's
theorem in alternative spacetime geometries, in particular the
`non-relativistic' and `ultra-relativistic' settings of Galilei and
Carroll geometries?\footnote{We use here the common designations
  `non-relativistic' and `ultra-relativistic' for Galilei and Carroll
  geometry, even though Galilean/Newtonian and Carrollian physics are
  relativistic as well, in the sense of satisfying the principle of
  relativity.  The principle is just implemented differently than in
  special or general relativity (or their modifications), namely by
  the Galilei or Carroll group, respectively, instead of the Poincaré
  group.}  It is upon such questions that we focus in the present
article, giving an affirmative answer for the case of Galilei
geometry, but showing that only a weaker result is provable in the
case of Carroll geometry.

In the Galilei case, there is some precedent for work on
`non-relativistic' versions of Weyl's theorem.  First,
\textcite{Ewen.Schmidt:1989} proved a version of this theorem---but
using a notion of Galilean conformal structure which, for reasons
discussed by \textcite{March.Read:2026, Adlam.Linnemann.Read:2025}, is
arguably conceptually inadequate.  (Essentially,
\citeauthor{Ewen.Schmidt:1989} define conformal structure as spacelike
projective structure, which is arguably insufficiently related to a
notion of Galilean scale transformations.)  Moreover,
\citeauthor{Ewen.Schmidt:1989} market their result as a Galilean
version of the construction of \textcite{EPS:1972, EPS:2012}, but as
pointed out by \textcite[ch.\ 4]{Adlam.Linnemann.Read:2025} this is to
confuse existence and uniqueness results of the kind discussed above.
(What \citeauthor{Ewen.Schmidt:1989} prove is an analogue of the 1921
uniqueness theorem of \citeauthor{Weyl:1921}, not the 1972 existence
result of \acronym{EPS}---see \citeauthor[ch.\
4]{Adlam.Linnemann.Read:2025} for further discussion.)  After having
presented a `non-relativistic' version of the \acronym{EPS}
construction using an improved notion of Galilean conformal structure
proposed by \textcite{March.Read:2026},\footnote{\label{fn:1}A precursor to
  this notion of Galilean conformal structure can be found in
  \textcite{Curiel:2015}, but it is more specialised than that
  proposed by \textcite{March.Read:2026}.  Curiel also claims to prove a
  Galilean version of Weyl's theorem, but the proof is erroneous as it
  omits crucial data needed to fix a connection uniquely with respect
  to a given Galilei geometry---namely, a choice of timelike vector
  field.  We discuss these extra data in detail in the body of this
  article.}  \textcite[\S\,4.5.1]{Adlam.Linnemann.Read:2025} assert
that for a `non-relativistic' version of Weyl's theorem `the proof is
exactly analogous to that sketched for the relativistic case', but do
not go into further details; these details are filled in by
\textcite{March:2025}, but only for Galilei structures on the
assumption of metric compatibility.\looseness-1

To expand on this final point, return for a moment to the Lorentzian
case.  In that setting, a special case of Weyl's result for Lorentzian
(rather than Weyl) metrics is presented by \textcite[ch.\
2]{Malament:2012}---this special case assumes metric compatibility.
Likewise, because \textcite{March:2025} assumes metric compatibility
in the Galilean context, what remains to be accomplished is a more
general (and directly analogous) version of Weyl's result for the
Galilean setting, in which one shows that a Galilean Weyl structure
(suitably defined---see \textcite[ch.\ 4]{Adlam.Linnemann.Read:2025}
and discussion below) is fixed uniquely by its associated projective
and conformal structures (plus, possibly, other relevant
data\footnote{See \cref{fn:1}.}).  Completing this work shores up and
completes our understanding of Galilean versions of Weyl's
theorem.\looseness-1

With this work on Galilean versions of Weyl's theorem complete, in
this article we in addition consider the case of Carroll versions of
the result.  (Recall that Carollian spacetime structures are the
`ultra-relativistic' versions of Lorentzian spacetime structures, in
which one sends $c \rightarrow 0$ and as such `narrows' the
lightcones---see \textcite{March.Read:2025} for a pedagogical review,
and \textcite{Vigneron.Barzegar.Read:2025} for results regarding
general Carrollian connections upon which we draw in this article.)
The conjecture here---in parallel with the Lorentzian and Galilean
cases---is that one can define a suitable notion of Carrollian
conformal structure such that a given Carrollian Weyl structure
(again, suitably defined) is fixed uniquely by its associated
projective and conformal structures (plus potential further relevant
data, as in the Galilean case).  In the second half of this article,
we show that a version of Weyl's theorem can indeed be proved in the
Carroll context, but the result is strictly weaker than the Lorentzian
and Galilean Weyl-type theorems: a Carrollian Weyl structure is
\emph{not} fixed uniquely by its associated projective and conformal
structures (plus potential further relevant data), although this
failure of uniqueness can be quantified in a precise
way.\footnote{\label{fn:Carr_constr_ax}We leave the project of a
  Carrollian version of the \acronym{EPS} construction---i.e.,\
  Carroll constructive axiomatics---for another day.  The failure of a
  Carroll version of Weyl's theorem might, however, be taken to
  indicate that there will exist serious roadblocks to such a
  construction.}

As such, the structure of this article is as follows.  In
\cref{sec:projective}, we review the notion of projective structure,
since this is in fact the same structure in all of the three cases
which we consider.  In \cref{sec:pseudo-Riemann}, we review the
definition of Lorentzian conformal structure and recall the essential
contours of the original theorem of \textcite{Weyl:1921}.  In
\cref{sec:Galilei}, we do the same for the Galilean case.  In
\cref{sec:Carroll}, we define Carrollian conformal structures and
prove a (weakened) Carrollian version of Weyl's theorem.  We close in
\cref{sec:concl} with some reflections on our results and potential
directions for future research.  (For the impatient, a concise summary
of our results may be found in \cref{tab:summary}.)  In
\cref{sec:app_class_conn_Carroll,sec:app_free_torsion_conf_Carroll},
we discuss some aspects of the classification of connections on
Carroll manifolds that are used in the main text.

\section{Projective structures}
\label{sec:projective}

In this section, we review some textbook material on projective
structures which will prove essential going forwards.  Further
discussion of projectively equivalent connections can be found in
e.g.\ \textcite{Malament:2012, Adlam.Linnemann.Read:2025}.

\begin{defn}
  Two linear connections on a differentiable manifold $M$ are
  \emph{projectively equivalent} if they have the same autoparallels
  (i.e.\ unparametrised geodesics).  A projective equivalence class of
  connections is a \emph{projective structure} on $M$.
\end{defn}

\begin{rem}
  Unlike much of the existing literature on projective structures,
  here we do \emph{not} demand the connections to be torsion-free.
  For this reason, we also provide the proof of the following standard
  result in the general case with possibly non-vanishing torsion:
\end{rem}

\begin{prop}
  Two linear connections on $M$ are projectively equivalent if and
  only if their difference tensor $D^\rho_{\mu\nu}$ satisfies
  \begin{equation}
    D^\rho_{(\mu\nu)}
    = \delta^\rho_{(\mu} \eta^{\vphantom{\rho}}_{\nu)}
  \end{equation}
  for some one-form $\eta \in \Omega^1(M)$.

  \begin{proof}
    A direct computation shows that if the difference tensor has the
    stated form, the two connections have the same autoparallels.

    For the converse, assume that two connections have the same
    autoparallels.  This implies that for any autoparallel $\gamma$ we
    have $D^\rho_{\mu\nu}(\gamma(t))) \dot\gamma^\mu(t)
    \dot\gamma^\nu(t) = f(t) \dot\gamma^\rho(t)$ for some function
    $f$.  Multiplying with $\dot\gamma^\sigma$ and antisymmetrising in
    $\rho,\sigma$, we obtain
    \begin{equation}
      0
      = \dot\gamma^{[\sigma\vphantom{\rho}}_{\vphantom{\mu}}
        D^{\rho]}_{\mu\nu} \dot\gamma^\mu \dot\gamma^\nu
      = D^{[\rho}_{\mu\nu} \delta^{\sigma]\vphantom{\rho}}_\kappa
        \dot\gamma^\mu \dot\gamma^\nu \dot\gamma^\kappa .
    \end{equation}
    Since any vector $v \in TM$ is tangent to some autoparallel, this
    shows that for all $v$ we have
    \begin{equation}
      D^{[\rho}_{\mu\nu} \delta^{\sigma]\vphantom{\rho}}_\kappa
        v^\mu v^\nu v^\kappa
      = 0,
    \end{equation}
    which is equivalent to
    \begin{equation}
      D^{[\rho}_{(\mu\nu} \delta^{\sigma]\vphantom{\rho}}_{\kappa)}
      = 0.
    \end{equation}

    Writing out the symmetrisation, this is equivalent to
    \begin{equation}
      D^{[\rho}_{(\mu\nu)} \delta^{\sigma]\vphantom{\rho}}_\kappa
      + D^{[\rho}_{(\mu\kappa)} \delta^{\sigma]\vphantom{\rho}}_\nu
      + D^{[\rho}_{(\nu\kappa)} \delta^{\sigma]\vphantom{\rho}}_\mu
      = 0.
    \end{equation}
    Contracting $\sigma$ and $\kappa$, we obtain
    \begin{align}
      0
      &= D^\rho_{(\mu\nu)} \delta^\sigma_\sigma
        - D^\sigma_{(\mu\nu)} \delta^\rho_\sigma
        + D^\rho_{(\mu\sigma)} \delta^\sigma_\nu
        - D^\sigma_{(\mu\sigma)} \delta^\rho_\nu
        + D^\rho_{(\nu\sigma)} \delta^\sigma_\mu
        - D^\sigma_{(\nu\sigma)} \delta^{\rho}_\mu \nonumber\\
      &= n D^\rho_{(\mu\nu)}
        - D^\rho_{(\mu\nu)}
        + D^\rho_{(\mu\nu)}
        - D^\sigma_{(\mu\sigma)} \delta^\rho_\nu
        + D^\rho_{(\mu\nu)}
        - D^\sigma_{(\nu\sigma)} \delta^\rho_\mu \nonumber\\
      &= (n+1) D^\rho_{(\mu\nu)}
        - D^\sigma_{(\mu\sigma)} \delta^\rho_\nu
        - D^\sigma_{(\nu\sigma)} \delta^\rho_\mu
    \end{align}
    where $n = \dim M$.  This shows that $D^\rho_{(\mu\nu)} =
    \delta^\rho_{(\mu} \eta^{\vphantom{\rho}}_{\nu)}$ with $\eta_\mu =
    \frac{2}{n+1} D^\sigma_{(\mu\sigma)}$.
  \end{proof}
\end{prop}

\section{The pseudo-Riemannian case}
\label{sec:pseudo-Riemann}

We next recall the definitions of conformal structures and Weyl
metrics in the pseudo-Riemannian setting (which includes as a special
case the Lorentzian, i.e.\ relativistic, one).  Having presented these
details, we then recall the above-mentioned theorem of
\textcite{Weyl:1921}.  All of this material is standard---see e.g.\
\textcite{Adlam.Linnemann.Read:2025} and references therein.

\begin{defn}
  Let $M$ be a differentiable manifold.
  \begin{enumerate}[label=(\roman*)]
  \item A \emph{pseudo-Riemannian conformal structure} on $M$ is an
    equivalence class $[g]$ of pseudo-Riemannian metrics on $M$ under
    the equivalence relation $g \sim \e^\lambda g$ for $\lambda \in
    C^\infty(M)$.
  \item A linear connection $\nabla$ on $M$ is \emph{compatible} with
    a pseudo-Riemannian conformal structure if for some (and then any)
    representative $g$ of the conformal structure there is a one-form
    $\varphi \in \Omega^1(M)$ such that $\nabla g = \varphi \otimes
    g$.
  \end{enumerate}
\end{defn}

\begin{rem}
  \begin{enumerate}[label=(\roman*)]
  \item If $\nabla g = \varphi \otimes g$, then for any conformally
    equivalent metric we have $\nabla(\e^\lambda g) = \e^\lambda
    \nabla g + \D\e^\lambda \otimes g = \varphi \otimes \e^\lambda g +
    \D\lambda \otimes \e^\lambda g = (\varphi + \D\lambda) \otimes
    \e^\lambda g$.  This shows the `and then any' part in the
    definition of compatible connections (and hence that the notion of
    compatibility is well-defined).
  \item Further, this shows that a pair $([g], \nabla)$ of a
    pseudo-Riemannian conformal structure and a compatible connection
    defines a specific equivalence class $[g,\varphi]$ of pairs of
    metrics and one-forms:
  \end{enumerate}
\end{rem}

\begin{defn}
  A \emph{Weyl metric} on a differentiable manifold $M$ is an
  equivalence class $[g,\varphi]$ of pairs of pseudo-Riemannian
  metrics and one-forms under the equivalence relation $(g, \varphi)
  \sim (\e^\lambda g, \varphi + \D\lambda)$ for $\lambda \in
  C^\infty(M)$.\footnote{Cf.\ the definition of a `Weyl class' by
    \textcite[p.\ 3]{Jia.Karydas.Leigh:2023}.}
\end{defn}

\begin{prop}
  Let $M$ be a differentiable manifold.  The natural map from the set
  of pairs of a pseudo-Riemannian conformal structure and a
  torsion-free compatible connection on $M$ to the set of Weyl metrics
  on $M$, mapping $([g], \nabla)$ to $[g, \varphi]$ where $\nabla g =
  \varphi \otimes g$, is a bijection.

  \begin{proof}
    We have seen above that a conformal structure and a compatible
    connection determine a well-defined Weyl metric.  Conversely, any
    representative $(g,\varphi)$ of a Weyl metric defines a unique
    torsion-free connection $\nabla$ satisfying $\nabla g = \varphi
    \otimes g$, since torsion-free connections are uniquely determined
    by their non-metricity with respect to a pseudo-Riemannian metric.
    By construction, this connection is compatible with $[g]$.
    Furthermore, the above computation shows that $\nabla$ is already
    determined by the equivalence class (i.e.\ the Weyl metric).
  \end{proof}
\end{prop}

\begin{thm}[\cite{Weyl:1921}]
  \label{thm:pseudo-Riemann_Weyl}
  Let $M$ be a differentiable manifold with $\dim M > 1$.  Let $[g]$
  be a pseudo-Riemannian conformal structure on $M$.  If two
  torsion-free linear connections $\nabla, \tilde{\nabla}$ are
  compatible with $[g]$ and projectively equivalent, then they agree,
  $\nabla = \tilde{\nabla}$.

  Put differently: torsion-free connections that are compatible with a
  given pseudo-Riemannian conformal structure are uniquely determined
  by their projective structure.

  \begin{proof}
    Fix some representative $g$ of the conformal structure.  By
    compatibility, there are one-forms $\varphi, \tilde{\varphi}$ such
    that
    \begin{equation}
      \nabla g = \varphi \otimes g,
      \tilde{\nabla} g = \tilde{\varphi} \otimes g.
    \end{equation}
    By projective equivalence, there is a one-form $\eta$ such that
    the difference tensor of the connections is given by
    \begin{equation}
      (\nabla - \tilde{\nabla})^\rho_{\mu\nu}
      = \delta^\rho_{(\mu} \eta^{\vphantom{\rho}}_{\nu)} \; .
    \end{equation}
    Note that due to torsion-freeness, this holds without
    symmetrisation of the difference in its lower indices.

    Combining these equations, we obtain
    \begin{align}
      \label{eq:Weyl_pf_1}
      \varphi_\rho g_{\mu\nu}
      &= \nabla_\rho g_{\mu\nu} \nonumber\\
      &= \tilde{\nabla}_\rho g_{\mu\nu}
        - (\nabla - \tilde{\nabla})^\kappa_{\rho\mu} g_{\kappa\nu}
        - (\nabla - \tilde{\nabla})^\kappa_{\rho\nu} g_{\mu\kappa}
        \nonumber\\
      &= \tilde{\varphi}_\rho g_{\mu\nu}
        - \delta^\kappa_{(\rho} \eta^{\vphantom{\kappa}}_{\mu)}
          g_{\kappa\nu}
        - \delta^\kappa_{(\rho} \eta^{\vphantom{\kappa}}_{\nu)}
          g_{\mu\kappa} \nonumber\\
      &= \tilde{\varphi}_\rho g_{\mu\nu}
        - \eta_{(\mu} g_{\rho)\nu}
        - g_{\mu(\rho} \eta_{\nu)} \nonumber\\
      &= \tilde{\varphi}_\rho g_{\mu\nu}
        - \eta_{\rho} g_{\mu\nu}
        - g_{\rho(\mu} \eta_{\nu)} \; .
    \end{align}
    Contracting \eqref{eq:Weyl_pf_1} with $g^{\mu\nu}$ yields
    \begin{equation}
      \label{eq:Weyl_pf_2}
      n \varphi_\rho
      = n \tilde{\varphi}_\rho - n \eta_\rho
        - g_{\rho\mu} g^{\mu\nu} \eta_\nu
      = n \tilde{\varphi}_\rho - (n+1) \eta_\rho \; ,
    \end{equation}
    where $n = \dim M$.  On the other hand, contracting
    \eqref{eq:Weyl_pf_1} with $g^{\mu\rho}$ yields
    \begin{equation}
      \label{eq:Weyl_pf_3}
      \varphi_\nu
      = \tilde{\varphi}_\nu - \eta_\nu
        - g^{\mu\rho} \frac{1}{2} (g_{\rho\mu} \eta_\nu
          + g_{\rho\nu} \eta_\mu)
      = \tilde{\varphi}_\nu - \eta_\nu - \frac{n+1}{2} \eta_\nu
      = \tilde{\varphi}_\nu - \frac{n+3}{2} \eta_\nu \; .
    \end{equation}

    Subtracting $n \cdot \eqref{eq:Weyl_pf_3}$ from
    \eqref{eq:Weyl_pf_2} (and renaming the free index), we have
    \begin{equation}
      0
      = -(n+1) \eta_\nu + \frac{n(n+3)}{2} \eta_\nu
      = \frac{n^2 + n - 2}{2} \eta_\nu \; .
    \end{equation}
    The expression $n^2 + n - 2$ vanishes if and only if $n \in \{-2,
    1\}$.  Thus, since we assumed $n > 1$, we obtain $\eta_\nu = 0$,
    finishing the proof.
  \end{proof}
\end{thm}

\begin{rem}
  \begin{enumerate}[label=(\roman*)]
  \item Since Weyl metrics are equivalently determined by their
    conformal structure and their associated torsion-free connection,
    Weyl's theorem may also be formulated as the statement that
    \emph{Weyl metrics are uniquely determined by their conformal and
      projective structures}.  This is the original formulation of the
    theorem as given by \citeauthor{Weyl:1921}.
  \item \textcite[Proposition 2.1.4]{Malament:2012} is a special case
    of this formulation of Weyl's theorem for pseudo-Riemannian
    metrics: \emph{pseudo-Riemannian metrics are determined by their
      conformal and projective structure up to multiplication by a
      locally constant function} (where the projective structure is
    that of the metric's Levi-Civita connection).  This follows from
    Weyl's theorem since the Levi-Civita connection $\nabla$ of $g$ is
    the torsion-free connection corresponding to the Weyl metric
    $[g,0]$, and an \emph{integrable} Weyl metric---i.e.\ one such
    that for all representatives $(\tilde{g}, \tilde{\varphi})$ the
    one-form $\tilde{\varphi}$ is exact---determines a
    pseudo-Riemannian metric up to a locally constant factor (namely
    any metric $\tilde{g}$ such that $(\tilde{g}, 0)$ is a
    representative).
  \item Note that for the equivalence between Weyl metrics and pairs
    of conformal structures and torsion-free compatible connections,
    it was crucial that a torsion-free connection be uniquely
    determined by its non-metricity.  If this were not the case,
    several connections could correspond to the same Weyl metric.
    However, \emph{Weyl's theorem does not need this one-to-one
      correspondence, as long as it is understood as a statement about
      uniqueness of connections}.  This will become important in the
    Galilei and Carroll cases: our Weyl-type theorems will be
    concerned with uniqueness of connections.  Only when adding extra
    data on top of the non-metricity form $\varphi$ (and part of the
    torsion, see below), they may be understood as statements on the
    uniqueness of (suitably defined) analogues of Weyl metrics.
  \end{enumerate}
\end{rem}

\begin{constr}
  We may include torsion into this picture as follows.  For two
  potentially \emph{torsionful} connections $\nabla$ and
  $\tilde{\nabla}$, if we assume their torsion to agree, $T =
  \tilde{T}$, this means that their difference tensor is symmetric,
  $(\nabla - \tilde{\nabla}) ^\rho_{\mu\nu} = (\nabla -
  \tilde{\nabla})^\rho_{(\mu\nu)}$.  Since the proof of Weyl's theorem
  (\cref{thm:pseudo-Riemann_Weyl}) as presented above used the
  assumption of torsion-freeness only for this equality, it applies in
  exactly the same manner to the case of two connections with equal
  torsion.

  Thus, we obtain the following torsionful generalisation of Weyl's
  theorem:
\end{constr}

\begin{thm}[Weyl's theorem with torsion]
  \label{thm:pseudo-Riemann_Weyl_torsion}
  Let $M$ be a differentiable manifold with $\dim M > 1$.  Let $[g]$
  be a pseudo-Riemannian conformal structure on $M$.  Let $\nabla,
  \tilde{\nabla}$ be two linear connections whose torsions agree, $T =
  \tilde{T}$.  If the connections are compatible with $[g]$ and
  projectively equivalent, then they agree, $\nabla = \tilde{\nabla}$.

  Put differently: connections that are compatible with a given
  pseudo-Riemannian conformal structure are uniquely determined by
  their torsion and their projective structure.  \qed
\end{thm}

\begin{rem}
  If we define a \emph{torsional Weyl metric} as a pair $([g,
  \varphi], T)$ of a usual Weyl metric $[g, \varphi]$ and a tensor
  field $T \in \Gamma(TM \otimes \bigwedge^2 T^*M)$, then these are in
  natural bijection to pairs $([g], \nabla)$ of pseudo-Riemannian
  conformal structures and compatible connections (possibly with
  torsion).  The above torsionful generalisation of Weyl's theorem may
  then be formulated as the statement that \emph{torsional Weyl
    metrics are uniquely determined by their torsion together with
    their conformal and projective structures}.
\end{rem}

\section{The Galilei case}
\label{sec:Galilei}

Our next task is to develop a parallel story to the above for the case
of Galilei geometry, on which see \textcite{Cartan:1923, Cartan:1924,
  Friedrichs:1928, Trautman:1963, Trautman:1965,
  Dombrowski.Horneffer:1964, Kuenzle:1972, Kuenzle:1976, Ehlers:1981a,
  Ehlers:1981b, Malament:2012, Hartong.EtAl:2023,
  Schwartz:NC_gravity}.  This builds upon prior work by
\textcite{Ewen.Schmidt:1989, Curiel:2015, March.Read:2026, March:2025,
  Adlam.Linnemann.Read:2025}.

\subsection{Prerequisites}

\begin{defn}
  Let $M$ be a differentiable manifold with $\dim M > 1$.
  \begin{enumerate}[label=(\roman*)]
  \item A \emph{Galilei structure}\footnote{We use here the name
      `Galilei structure' for the pair of two `metric' tensor fields
      of a Galilei manifold, while this term commonly refers to the
      reduction of the structure group of the frame bundle induced by
      them (i.e.\ the Galilei frame bundle).  Since we make no use of
      frame bundles in this article, this should not cause confusion.}
    on $M$ is a pair $(\tau,h)$ of a nowhere-vanishing \emph{clock
      form} $\tau \in \Omega^1(M)$ and a \emph{space metric} $h \in
    \Gamma(\bigvee^2 TM)$ which is positive semidefinite of rank $n =
    \dim M - 1$ such that $\tau$ spans the degenerate direction of
    $h$, i.e.
    \begin{equation}
      \tau_\mu h^{\mu\nu} = 0.
    \end{equation}    
  \item A \emph{conformal Galilei structure} on $M$ is an equivalence
    class $[\tau,h]$ of Galilei structures under the equivalence
    relation $(\tau,h) \sim (\e^{\frac{1}{2}\lambda} \tau,
    \e^{-\lambda} h)$ for $\lambda \in C^\infty(M)$.
  \item A linear connection $\nabla$ on $M$ is \emph{compatible} with
    a conformal Galilei structure if for some (and then any)
    representative $(\tau,h)$ of the conformal structure there is a
    one-form $\varphi \in \Omega^1(M)$ such that $\nabla \tau =
    \frac{1}{2} \varphi \otimes \tau$, $\nabla h = - \varphi \otimes
    h$.
  \end{enumerate}
\end{defn}

\begin{rem}
  \begin{enumerate}[label=(\roman*)]
  \item We chose here a definition of conformal equivalence that uses
    related conformal factors for the clock form and space metric, as
    it would result from a Newtonian limit of a Lorentzian conformal
    transformation.  This is also the reason for the definition of
    compatibility using related one-forms.

    One could try to relax this relation (i.e.\ use independent
    conformal equivalence classes of clock forms and space
    metrics---see \textcite{Curiel:2015, March.Read:2026,
      Adlam.Linnemann.Read:2025}).  However, we try to stay as close
    as possible to the Lorentzian origin of Galilei geometry here.
    Furthermore, our Weyl-type theorem needs this stronger notion of
    conformal structure.
  \item As in the pseudo-Riemannian case, the one-forms encoding the
    non-metricity of a compatible connection with respect to different
    representatives conformally related by $\lambda$ are related by
    $\varphi \mapsto \varphi + \D\lambda$.
  \item In Galilei geometry, the non-metricities of a connection need
    to satisfy two conditions (\cite{Schwartz:2025}).  The one
    relating only the non-metricities, namely $\tau_\mu \nabla_\rho
    h^{\mu\nu} = - h^{\mu\nu} \nabla_\rho \tau_\mu$, is satisfied for
    the non-metricities from the definition of compatibility.  The
    second condition involves the torsion: it reads
    \begin{subequations}
    \begin{align}
      \tau_\rho \tensor{T}{^\rho_{\mu\nu}}
      & = (\D\tau)_{\mu\nu} - 2 \nabla_{[\mu} \tau_{\nu]} \; . \\
      \intertext{For the above non-metricity, this becomes}
      \label{eq:Gal_conf_compat_torsion}
      \tau(T(\cdot, \cdot))
      &= \D\tau - \frac{1}{2} \varphi \wedge \tau .
    \end{align}
    \end{subequations}
    This shows that connections compatible with a conformal Galilei
    structure can only be torsion-free if $\tau$ satisfies the
    Frobenius integrability condition, and even then the non-metricity
    form $\varphi$ is related to $\D\tau$.  Thus in general compatible
    connections are not torsion-free.
  \end{enumerate}
\end{rem}

Since---differently to the pseudo-Riemannian case---we cannot in
general demand that compatible connections be torsion-free, we need to
allow for torsion.  However, we can hope for the uniqueness of
connections in a Galilean Weyl-type theorem only if we (partially)
restrict the torsion, similar to the torsionful version of Weyl's
theorem in the pseudo-Riemannian case
(\cref{thm:pseudo-Riemann_Weyl_torsion}).  We are going to fix the
`spatial' part of the torsion with respect to some specified timelike
direction, which, according to the classification result of
\textcite{Schwartz:2025}, is precisely the freely specifiable part of
the torsion:

\begin{defn}
  Given a conformal Galilei structure $[\tau, h]$ on a differentiable
  manifold $M$, the notion of a \emph{timelike direction} (i.e.\ a
  timelike vector field up to rescaling) with respect to $[\tau,h]$ is
  well-defined: it is an equivalence class $[v]$ of timelike vector
  fields, $\tau(v) \ne 0$, under the equivalence relation $v \sim
  \tilde{\lambda} v$ for $\tilde{\lambda} \in C^\infty(M, \R \setminus
  \{0\})$.  Given a timelike direction $[v]$, there is a well-defined
  unique \emph{spatial projector along $[v]$}, the vector bundle
  endomorphism $P$ of $TM$ defined by projecting onto $\ker\tau$
  (which is invariant under rescaling of $\tau$) and having kernel
  spanned by $v$ (which is invariant under rescaling of $v$).

  Fixing a representative $(\tau, h)$ of the conformal
  structure, there is a unique representative $v$ of the timelike
  direction satisfying $\tau(v) = 1$, and in terms of these the
  spatial projector takes the usual explicit form $P^\mu_\nu =
  \delta^\mu_\nu - v^\mu \tau_\nu$.
\end{defn}

\subsection{The Galilei Weyl-type theorem}

With these prerequisites, we can now prove a Weyl-type theorem:

\begin{thm}
  \label{thm:Gal_Weyl}
  Let $M$ be a differentiable manifold with $\dim M > 2$.  Let $[\tau,
  h]$ be a conformal Galilei structure on $M$, and $[v]$ a timelike
  direction with respect to $[\tau, h]$.  Let $\nabla, \tilde{\nabla}$
  be two linear connections whose spatial torsions with respect to
  $[v]$ agree, i.e.\ their torsions satisfy $P(T(\cdot,\cdot)) =
  P(\tilde{T}(\cdot, \cdot))$ where $P$ is the spatial projector along
  $[v]$.  If the connections are compatible with $[\tau, h]$ and
  projectively equivalent, then they agree, $\nabla = \tilde{\nabla}$.

  \begin{rem}
    Note that for this theorem to apply, we only need to assume
    agreement of the spatial torsions of the two connections with
    respect to some \emph{specific choice} of timelike direction
    $[v]$.  Of course, from the theorem's conclusion that the
    connections are equal, it follows that the spatial torsions agree
    with respect to \emph{every} spatial direction; however, agreement
    for \emph{one} $[v]$ is sufficient for this conclusion to hold.
  \end{rem}

  \begin{proof}[Proof of \cref{thm:Gal_Weyl}]
    Fix some representative $(\tau,h)$ of the conformal structure, and
    let $v$ be the corresponding representative of the timelike
    direction satisfying $\tau(v) = 1$.  Denote the difference tensor
    of the connections by $D := \nabla - \tilde{\nabla}$.  

    By projective equivalence, there is a one-form $\eta$ such that
    the difference tensor satisfies
    \begin{equation}
      \label{eq:Gal_Weyl_pf_proj_eq}
      D^\rho_{(\mu\nu)}
      = \delta^\rho_{(\mu} \eta^{\vphantom{\rho}}_{\nu)} \; .
    \end{equation}
    Furthermore, both connections having the same spatial torsion with
    respect to $[v]$ is equivalent to
    \begin{subequations}
    \begin{align}
      P^\sigma_\rho D^\rho_{[\mu\nu]}
      &= 0, \\
      \shortintertext{implying}
      \label{eq:Gal_Weyl_pf_diff_proj}
      P^\sigma_\rho D^\rho_{\mu\nu}
      &= P^\sigma_\rho D^\rho_{(\mu\nu)}
        = P^\sigma_{(\mu} \eta^{\vphantom{\sigma}}_{\nu)} \; .
    \end{align}
    \end{subequations}

    By compatibility of the connections with $[\tau, h]$, there are
    one-forms $\varphi, \tilde{\varphi}$ such that
    \begin{equation}
      \nabla \tau = \frac{1}{2} \varphi \otimes \tau,
      \tilde{\nabla} \tau = \frac{1}{2} \tilde{\varphi} \otimes \tau,
      \nabla h = - \varphi \otimes h,
      \tilde{\nabla} h = - \tilde{\varphi} \otimes h.
    \end{equation}

    Combining the two equations for compatibility with $\tau$, we
    obtain
    \begin{align}
      \label{eq:Gal_Weyl_pf_compat_tau}
      \frac{1}{2} \varphi_\mu \tau_\nu
      &= \nabla_\mu \tau_\nu \nonumber\\
      &= \tilde{\nabla}_\mu \tau_\nu - D^\rho_{\mu\nu} \tau_\rho
        \nonumber\\
      &= \frac{1}{2} \tilde{\varphi}_\mu \tau_\nu
        - D^\rho_{\mu\nu} \tau_\rho \; .
    \end{align}
    Contracting with $h^{\nu\sigma}$, from this we obtain
    \begin{equation}
      \label{eq:Gal_Weyl_pf_diff_tau_h}
      \tau_\rho D^\rho_{\mu\nu} h^{\nu\sigma} = 0.
    \end{equation}

    Symmetrising the free indices of \eqref{eq:Gal_Weyl_pf_compat_tau}
    yields
    \begin{align}
      \frac{1}{2} \varphi_{(\mu} \tau_{\nu)}
      &= \frac{1}{2} \tilde{\varphi}_{(\mu} \tau_{\nu)}
        - D^\rho_{(\mu\nu)} \tau_\rho \nonumber\\
      &= \frac{1}{2} \tilde{\varphi}_{(\mu} \tau_{\nu)}
        - \tau_{(\mu} \eta_{\nu)} \; ,
    \end{align}
    where we have used \eqref{eq:Gal_Weyl_pf_proj_eq} for the
    symmetrised difference tensor.  Contracting this with $v^\mu
    v^\nu$, we obtain
    \begin{equation}
      \label{eq:Gal_Weyl_pf_1}
      \frac{1}{2} \varphi(v)
      = \frac{1}{2} \tilde{\varphi}(v) - \eta(v).
    \end{equation}

    On the other hand, combining the two equations for compatibility
    with $h$, we obtain
    \begin{align}
      \label{eq:Gal_Weyl_pf_compat_h_1}
      -\varphi_\rho h^{\mu\nu}
      &= \nabla_\rho h^{\mu\nu} \nonumber\\
      &= \tilde{\nabla}_\rho h^{\mu\nu}
        + D^\mu_{\rho\kappa} h^{\kappa\nu}
        + D^\nu_{\rho\kappa} h^{\mu\kappa} \nonumber\\
      &= -\tilde{\varphi}_\rho h^{\mu\nu}
        + 2 D^{(\mu}_{\rho\kappa}
          h^{\nu)\kappa\vphantom{\mu}}_{\vphantom{\rho}} .
    \end{align}
    In order to simplify this equation, we consider its last term.
    From \eqref{eq:Gal_Weyl_pf_diff_tau_h}, we know that the
    \pagebreak
    expression $D^\rho_{\mu\nu} h^{\nu\sigma}$ vanishes when
    contracted with $\tau_\rho$; i.e.\ the index $\rho$ is spacelike.
    Therefore, the expression is equal to its projection onto space
    along $[v]$ on this index,
    \begin{subequations}
    \begin{align}
      D^\rho_{\mu\nu} h^{\nu\sigma}
      &= P^\rho_\lambda D^\lambda_{\mu\nu} h^{\nu\sigma} . \\
      \intertext{But the spacelike projected difference tensor can be
      expressed in terms of $\eta$ according to
      \eqref{eq:Gal_Weyl_pf_diff_proj},
      giving}
      \label{eq:Gal_Weyl_pf_diff_h}
      D^\rho_{\mu\nu} h^{\nu\sigma}
      &= P^\rho_{(\mu} \eta^{\vphantom{\rho}}_{\nu)} h^{\nu\sigma}
        \nonumber\\
      &= \frac{1}{2} (P^\rho_\mu \eta_\nu h^{\nu\sigma}
        + P^\rho_\nu \eta_\mu h^{\nu\sigma}) \nonumber\\
      &= \frac{1}{2} (P^\rho_\mu \eta^\sigma
        + \eta_\mu h^{\rho\sigma}) \; ,
    \end{align}
    \end{subequations}
    where $\eta^\mu := h^{\mu\nu} \eta_\nu$ is the spacelike vector
    field corresponding to $\eta$.  Now using
    \eqref{eq:Gal_Weyl_pf_diff_h}, the compatibility equation
    \eqref{eq:Gal_Weyl_pf_compat_h_1} takes the form
    \begin{equation}
      \label{eq:Gal_Weyl_pf_compat_h_2}
      -\varphi_\rho h^{\mu\nu}
      = -\tilde{\varphi}_\rho h^{\mu\nu}
        + P^{(\mu}_\rho \eta^{\nu)\vphantom{\mu}}_{\vphantom{\rho}}
        + \eta_\rho h^{\mu\nu} .
    \end{equation}
    Contracting this with $v^\rho$ gives
    \begin{subequations}
    \begin{align}
      -\varphi(v) h^{\mu\nu}
      &= -\tilde{\varphi}(v) h^{\mu\nu} + \eta(v) h^{\mu\nu} , \\
      \shortintertext{i.e.}
      \label{eq:Gal_Weyl_pf_2}
      -\varphi(v)
      &= -\tilde{\varphi}(v) + \eta(v).
    \end{align}
    \end{subequations}
    Adding $2 \cdot \eqref{eq:Gal_Weyl_pf_1}$ to
    \eqref{eq:Gal_Weyl_pf_2}, we obtain
    \begin{equation}
      \label{eq:Gal_Weyl_pf_eta_v_0}
      0 = -\eta(v).
    \end{equation}

    \enlargethispage{2\baselineskip}

    Contracting \eqref{eq:Gal_Weyl_pf_compat_h_2} with
    $\gamma_{\mu\nu}$---the components of the covariant space
    metric\footnote{In the literature on Galilei geometry /
      Newton--Cartan gravity, the components of the covariant space
      metric are commonly denoted by $h_{\mu\nu}$.  We choose here the
      symbol $\gamma$ instead, in parallel to Carroll geometry; cf.\
      \textcite{Vigneron.Barzegar.Read:2025}.}  with respect to $v$,
    defined by the requirements $\gamma_{\mu\nu} = \gamma_{\nu\mu}$,
    $\gamma_{\mu\nu} v^\nu = 0$, and $\gamma_{\mu\nu} h^{\nu\rho} =
    P^\rho_\mu$---yields
    \begin{subequations}
    \begin{align}
      -n \varphi_\rho
      &= -n \tilde{\varphi}_\rho + P^\mu_\rho \eta_\mu
        + n \eta_\rho \; , \\
      \intertext{where $n = \dim M - 1$.  Contracting this with
      $h^{\rho\nu}$, we obtain}
      \label{eq:Gal_Weyl_pf_3}
      -n \varphi^\nu
      &= -n \tilde{\varphi}^\nu + (n+1) \eta^\nu .
    \end{align}
    \end{subequations}
    On the other hand, contracting
    \eqref{eq:Gal_Weyl_pf_compat_h_2} with $\delta^\rho_\mu$ yields
    \begin{equation}
      \label{eq:Gal_Weyl_pf_4}
      -\varphi^\nu
      = -\tilde{\varphi}^\nu
        + \delta^\rho_\mu \frac{1}{2} (P^\mu_\rho \eta^\nu
          + P^\nu_\rho \eta^\mu)
        + \eta^\nu
      = -\tilde{\varphi}^\nu + \frac{n+1}{2} \eta^\nu + \eta^\nu
      = -\tilde{\varphi}^\nu + \frac{n+3}{2} \eta^\nu .
    \end{equation}
    Subtracting \eqref{eq:Gal_Weyl_pf_3} from $n \cdot
    \eqref{eq:Gal_Weyl_pf_4}$, we have
    \begin{equation}
      0
      = \frac{n(n+3)}{2} \eta^\nu - (n+1) \eta^\nu
      = \frac{n^2 + n - 2}{2} \eta^\nu .
    \end{equation}
    The expression $n^2 + n - 2$ vanishes if and only if $n \in \{-2,
    1\}$.  Thus, since we assumed $n > 1$, we obtain
    \begin{equation}
      \eta^\nu = 0.
    \end{equation}
    Since we can express $\eta$ as $\eta_\mu = (v^\rho \tau_\mu +
    P^\rho_\mu) \eta_\rho = \eta(v) \tau_\mu + \gamma_{\mu\nu}
    h^{\nu\rho} \eta_\rho = \eta(v) \tau_\mu + \gamma_{\mu\nu}
    \eta^\nu$, together with \eqref{eq:Gal_Weyl_pf_eta_v_0} this shows
    $\eta = 0$.  Hence the symmetrised part of the difference tensor
    vanishes,
    \begin{equation}
      D^\rho_{(\mu\nu)} = 0,
    \end{equation}
    and due to \eqref{eq:Gal_Weyl_pf_compat_h_2} we have
    \begin{equation}
      \varphi = \tilde{\varphi}.
    \end{equation}

    The antisymmetric part of the difference $D$ is given by (one half
    of) the difference of the torsions of the two connections,
    \begin{subequations}
    \begin{align}
      D^\rho_{[\mu\nu]}
      &= \frac{1}{2} (\tensor{T}{^\rho_{\mu\nu}}
        - \tensor{\tilde{T}}{^\rho_{\mu\nu}}\!). \\
      \intertext{Due to the assumption of agreeing spatial torsion
      with respect to $[v]$, this is given by}
      \label{eq:Gal_Weyl_pf_D_antisymm}
      D^\rho_{[\mu\nu]}
      &= \frac{1}{2} v^\rho \tau_\sigma
        (\tensor{T}{^\sigma_{\mu\nu}}
        - \tensor{\tilde{T}}{^\sigma_{\mu\nu}}\!),
    \end{align}
    \end{subequations}
    i.e.\ by the difference of the temporal torsions of the two
    connections.  However, according to
    \eqref{eq:Gal_conf_compat_torsion} the temporal torsions may be
    expressed in terms of $\D\tau$ and $\varphi, \tilde{\varphi}$.
    Thus, due to $\varphi = \tilde{\varphi}$, the temporal torsions of
    $\nabla$ and $\tilde{\nabla}$ agree.  This shows $D = 0$,
    finishing the proof.
  \end{proof}
\end{thm}

\begin{rem}
  Note that the proof of the vanishing of the spacelike part of $\eta$
  was directly analogous to the pseudo-Riemannian case.  In fact, if
  we had assumed $\tau \wedge \D\tau = 0$ (which is invariant under
  conformal rescaling of $\tau$), we could have appealed directly to
  the Riemannian version of Weyl's theorem for this conclusion: $\tau
  \wedge \D\tau = 0$ is equivalent to the spacelike distribution $\ker
  \tau$ being integrable by $n$-dimensional `spatial leaves' $\Sigma$,
  on each of which $[h]$ then induces a Riemannian conformal
  structure.  The connections $\nabla, \tilde{\nabla}$ then induce
  `spatial connections' on these leaves, both of which are compatible
  with the Riemannian conformal structure.  Further, the spatial
  connections' geodesics are precisely the spacelike geodesics of the
  spacetime connections, which agree up to parametrisation; i.e.\ on
  each spatial leaf the two spatial connections are projectively
  equivalent.  Thus, Weyl's theorem implies that the spatial
  connections agree.
\end{rem}

\subsection{A Galilei analogue of Weyl metrics}

We may define an analogue of Weyl metrics for Galilei geometry as
follows.

\enlargethispage{\baselineskip}

\begin{constr}
  Fixing a Galilei structure $(\tau, h)$ on $M$, with respect to a
  choice of unit timelike vector field $v \in \Gamma(TM)$, $\tau(v) =
  1$, any linear connection $\nabla$ on $M$ is uniquely determined
  by---and uniquely determines, vice versa---its $\tau$-non-metricity
  $\nabla \tau$, the spatial part $P^\rho_\mu P^\sigma_\nu
  \nabla_\kappa h^{\mu\nu}$ of its $h$-non-metricity, its spatial
  torsion $\tensor{(PT)}{^\sigma_{\mu\nu}} := P^\sigma_\rho
  \tensor{T}{^\rho_{\mu\nu}}$, and its Newton--Coriolis form
  $\Omega_{\mu\nu} = 2 (\nabla_{[\mu} v^\rho) \gamma_{\nu]\rho}$; see
  \textcite{Schwartz:2025}.  Assuming the non-metricities to be of the
  form $\nabla \tau = \frac{1}{2} \varphi \otimes \tau$, $\nabla h =
  -\varphi \otimes h$, such that $\nabla$ is compatible with the
  conformal Galilei structure $[\tau, h]$, we already know that fixing
  the connection $\nabla$, under a conformal rescaling $(\tau,h)
  \mapsto (\e^{\frac{1}{2}\lambda} \tau, \e^{-\lambda} h)$ with
  $\lambda \in C^\infty(M)$ the one-form $\varphi$ transforms as
  $\varphi \mapsto \varphi + \D\lambda$.  Further, fixing the timelike
  direction $[v]$, the spatial projector $P$ along it stays the same,
  such that the spatial torsion $PT$ stays invariant.  Rescaling $v
  \mapsto \e^{-\frac{1}{2}\lambda}$ such that it remains unit timelike
  with respect to the new representative of the conformal structure, a
  direct computation shows that the Newton--Coriolis form of $\nabla$
  with respect to $v$ transforms as $\Omega \mapsto
  \e^{-\frac{1}{2}\lambda} \Omega$.

  Therefore, a conformal Galilei structure $[\tau, h]$ on $M$ together
  with a compatible linear connection is equivalently characterised by
  an equivalence class $[\tau, h, \varphi, v, PT, \Omega]$ of Galilei
  structures $(\tau, h)$, one-forms $\varphi$, unit timelike vector
  fields $v$ with respect to $\tau$, tensor fields $PT \in
  \Gamma(\ker\tau \otimes \bigwedge^2 T^*M)$, and two-forms $\Omega$
  on $M$ under the equivalence relation
  \begin{equation}
    (\tau, h, \varphi, v, PT, \Omega)
    \sim (\e^{\frac{1}{2}\lambda} \tau, \e^{-\lambda} h,
      \varphi + \D\lambda, \e^{-\frac{1}{2}\lambda} v, PT,
      \e^{-\frac{1}{2}\lambda} \Omega)
  \end{equation}
  for $\lambda \in C^\infty(M)$.

  Alternatively, since $v$ is uniquely determined by the timelike
  direction $[v]$ once $\tau$ is fixed, and $PT$ is invariant, these
  are the same data as a triple $([\tau, h, \varphi, \Omega], [v],
  PT)$ consisting of (i)~an equivalence class of Galilei structures,
  one-forms, and two-forms on $M$ under
  \begin{equation}
    (\tau, h, \varphi, \Omega)
    \sim (\e^{\frac{1}{2}\lambda} \tau, \e^{-\lambda} h,
      \varphi + \D\lambda, \e^{-\frac{1}{2}\lambda} \Omega);
  \end{equation}
  (ii)~a timelike direction $[v]$ with respect to the conformal
  Galilei structure $[\tau, h]$; and (iii)~a tensor field $PT \in
  \Gamma(\ker\tau \otimes \bigwedge^2 T^*M)$.

  Thus, we may reformulate our Weyl-type theorem for Galilei geometry
  (\cref{thm:Gal_Weyl}) as the statement that \emph{such a `Galilei
    Weyl metric' $([\tau, h, \varphi, \Omega], [v], PT)$ is uniquely
    determined by its conformal and projective structures together
    with its timelike direction $[v]$ and its spatial torsion $PT$
    with respect to $[v]$}.

  Note that differently to the pseudo-Riemannian case, for a `Galilei
  Weyl metric' we always include the freely specifiable torsion $PT$
  as part of the data, since its vanishing is not a meaningful
  invariant statement: the same connection may have $PT = 0$ with
  respect to some choice $[v]$ of timelike direction, but $\tilde{P}T
  \ne 0$ with respect to some other direction $[\tilde{v}]$.\looseness-1
\end{constr}

\begin{rem}
  The above formulation of the notion of a `Galilei Weyl metric'
  corrects and generalises the notion of an `\acronym{NR}--Weyl
  manifold' presented by \textcite[p.\ 170]
  {Adlam.Linnemann.Read:2025}, in which the data $v, PT, \Omega$ are
  not included.
\end{rem}

\subsection{Necessity of the dimensionality assumption}

We are now going to show that the `additional' assumption in
\cref{thm:Gal_Weyl} as compared to the torsionful version of Weyl's
theorem in the pseudo-Riemannian case, namely that $\dim M > 2$, is
necessary for the conclusion of agreeing connections: we are going to
show that \cref{thm:Gal_Weyl} fails to apply if $\dim M = 2$, and in
fact precisely classify \emph{how} it fails.

\begin{prop}
  \label{prop:Gal_diff_dim_2}
  Let $M$ be a differentiable manifold with $\dim M = 2$.  Let $[\tau,
  h]$ be a conformal Galilei structure on $M$, and $[v]$ a timelike
  direction with respect to $[\tau, h]$.

  \begin{enumerate}[label=(\roman*)]
  \item Let $\nabla, \tilde{\nabla}$ be two linear connections whose
    spatial torsions with respect to $[v]$ agree, i.e.\ their torsions
    satisfy $P(T(\cdot,\cdot)) = P(\tilde{T}(\cdot, \cdot))$ where $P$
    is the spatial projector along $[v]$.

    If the connections are compatible with $[\tau, h]$ and
    projectively equivalent, then their difference tensor $D := \nabla
    - \tilde{\nabla}$ is given by
    \begin{subequations}
    \begin{equation}
      \label{eq:Gal_diff_dim_2}
      D^\rho_{\mu\nu}
      = \eta_\mu \delta^\rho_\nu
        - \eta^{\vphantom{\rho}}_{[\mu} P^\rho_{\nu]}
    \end{equation}
    for a one-form $\eta \in \Omega^1(M)$ that is spacelike with
    respect to $[v]$, i.e.\ that satisfies
    \begin{equation}
      \eta(v) = 0.
    \end{equation}
    \end{subequations}
  \item Conversely, let $\nabla$ be a linear connection that is
    compatible with $[\tau, h]$.  Then for \emph{any} choice of $\eta$
    with $\eta(v) = 0$, the connection $\tilde{\nabla} := \nabla - D$
    with difference $D$ given by \eqref{eq:Gal_diff_dim_2} is
    compatible with $[\tau, h]$ and has the same spatial torsion with
    respect to $[v]$ as $\nabla$, and the connections are projectively
    equivalent.
  \end{enumerate}

  \begin{proof}
    Fix some representative $(\tau,h)$ of the conformal structure, and
    let $v$ be the corresponding representative of the timelike
    direction satisfying $\tau(v) = 1$.

    \begin{enumerate}[label=(\roman*)]
    \item As in the proof of \cref{thm:Gal_Weyl}, we denote the
      one-form determining the symmetric part of the difference tensor
      $D = \nabla - \tilde{\nabla}$ by $\eta$, and the one-forms
      determining the non-metricities of $\nabla$ and $\tilde{\nabla}$
      by $\varphi$ and $\tilde{\varphi}$, respectively.

      In the proof of \cref{thm:Gal_Weyl}, we did not use any
      assumption on $\dim M$ in order to show that
      \begin{subequations}
      \begin{equation}
        \eta(v) = 0 \; \text{and} \; \varphi(v) = \tilde{\varphi}(v).
      \end{equation}
      Therefore, these still hold in the present case $\dim M = 2$.
      However, the spatial part $\eta^\mu = h^{\mu\nu} \eta_\nu$ of
      $\eta$ need not vanish in the case $n = \dim M - 1 = 1$: both
      \eqref{eq:Gal_Weyl_pf_3} and \eqref{eq:Gal_Weyl_pf_4} become
      \begin{equation}
        \varphi^\mu = \tilde{\varphi}^\mu - 2 \eta^\mu
      \end{equation}
      \end{subequations}
      in this case.  Together, these equations show that for $\dim M =
      2$ we have
      \begin{equation}
        \label{eq:Gal_phi_diff_dim_2}
        \varphi = \tilde{\varphi} - 2 \eta,
      \end{equation}
      with $\eta$ spacelike with respect to $[v]$.

      Further, using \eqref{eq:Gal_conf_compat_torsion} this implies
      that the difference of the temporal torsions of $\nabla$ and
      $\tilde{\nabla}$ is given by
      \begin{equation}
        \tau_\rho (\tensor{T}{^\rho_{\mu\nu}}
          - \tensor{\tilde{T}}{^\rho_{\mu\nu}}\!)
        = -(\varphi_{[\mu} - \tilde{\varphi}_{[\mu}) \tau_{\nu]}
        = 2 \eta_{[\mu} \tau_{\nu]} \; .
      \end{equation}
      Therefore, using \eqref{eq:Gal_Weyl_pf_D_antisymm}, the
      antisymmetric part of the difference tensor is
      $D^\rho_{[\mu\nu]} = v^\rho \eta_{[\mu} \tau_{\nu]}$.  Together
      with the equation determining the symmetric part of $D$ in terms
      of $\eta$, we thus obtain
      \begin{equation}
        D^\rho_{\mu\nu}
        = \delta^\rho_{(\mu} \eta^{\vphantom{\rho}}_{\nu)\vphantom{\mu}}
          + \eta_{[\mu} v^\rho \tau_{\nu]}
        = \delta^\rho_{(\mu} \eta^{\vphantom{\rho}}_{\nu)\vphantom{\mu}}
          + \eta^{\vphantom{\rho}}_{[\mu} (\delta^\rho_{\nu]}
            - P^\rho_{\nu]})
      = \eta_\mu \delta^\rho_\nu
        - \eta^{\vphantom{\rho}}_{[\mu} P^\rho_{\nu]}
      \end{equation}
      as stated.

    \item Again, we denote the one-form determining the
      non-metricities of $\nabla$ by $\varphi$, such that we have
      \begin{equation}
        \nabla \tau = \frac{1}{2} \varphi \otimes \tau,
        \nabla h = - \varphi \otimes h.
      \end{equation}
      Let $\eta \in \Omega^1(M)$ be an arbitrary one-form satisfying
      $\eta(v) = 0$, and define $\tilde{\nabla} := \nabla - D$ with
      $D$ given by \eqref{eq:Gal_diff_dim_2}.

      Directly from \eqref{eq:Gal_diff_dim_2}, we obtain
      \begin{equation}
        D^\rho_{(\mu\nu)}
        = \delta^\rho_{(\mu} \eta^{\vphantom{\rho}}_{\nu)} \; ,
      \end{equation}
      showing projective equivalence of $\nabla$ and $\tilde{\nabla}$.
      Further, the difference of their spatial torsions with respect
      to $[v]$ is
      \begin{equation}
        P^\sigma_\rho D^\rho_{[\mu\nu]}
        = P^\sigma_\rho \eta^{\vphantom{\rho}}_{[\mu}
          (\delta^\rho_{\nu]} - P^\rho_{\nu]})
        = P^\sigma_\rho \eta_{[\mu} v^\rho \tau_{\nu]}
        = 0,
      \end{equation}
      i.e.\ the spatial torsions agree.

      Using \eqref{eq:Gal_diff_dim_2}, we may compute
      \begin{subequations}
      \begin{align}
        \tilde{\nabla}_\mu \tau_\nu
        &= \nabla_\mu \tau_\nu + D^\rho_{\mu\nu} \tau_\rho \nonumber\\
        &= \frac{1}{2} \varphi_\mu \tau_\nu
          + (\eta_\mu \delta^\rho_\nu
            - \eta^{\vphantom{\rho}}_{[\mu} P^\rho_{\nu]}) \tau_\rho
          \nonumber\\
        &= \frac{1}{2} (\varphi_\mu + 2 \eta_\mu) \tau_\nu \\
        \shortintertext{and}
        \tilde{\nabla}_\mu h^{\rho\sigma}
        &= \nabla_\mu h^{\rho\sigma} - D^\rho_{\mu\nu} h^{\nu\sigma}
          - D^\sigma_{\mu\nu} h^{\rho\nu} \nonumber\\
        &= - \varphi_\mu h^{\rho\sigma} - (\eta_\mu \delta^\rho_\nu
              - \eta^{\vphantom{\rho}}_{[\mu} P^\rho_{\nu]})
            h^{\nu\sigma}
          - (\eta_\mu \delta^\sigma_\nu
              - \eta^{\vphantom{\sigma}}_{[\mu} P^\sigma_{\nu]})
            h^{\rho\nu} \nonumber\\
        &= - \varphi_\mu h^{\rho\sigma}
          - \eta^{\vphantom{\rho}}_{(\mu} P^\rho_{\nu)} h^{\nu\sigma}
          - \eta^{\vphantom{\sigma}}_{(\mu} P^\sigma_{\nu)}
            h^{\rho\nu} \nonumber\\
        &= - \varphi_\mu h^{\rho\sigma}
          - \frac{1}{2} (\eta_\mu h^{\rho\sigma}
            + \eta^\sigma P^\rho_\mu)
          - \frac{1}{2} (\eta_\mu h^{\rho\sigma}
            + \eta^\rho P^\sigma_\mu) \nonumber\\
        &= -(\varphi_\mu + \eta_\mu) h^{\rho\sigma}
          - P^{(\rho}_\mu \eta^{\sigma)\vphantom{\rho}}_{\vphantom{\mu}}
          \; . \\
        \intertext{Due to $\dim M = 2$, i.e.\ there being only one
        spatial direction, and using that $\eta$ is spacelike, we have
        $P^{(\rho}_\mu \eta^{\sigma)\vphantom{\rho}}_{\vphantom{\mu}}
        = \eta_\mu h^{\rho\sigma}$.  Hence, the previous equation
        becomes}
        \tilde{\nabla}_\mu h^{\rho\sigma}
        &= -(\varphi_\mu + 2 \eta_\mu) h^{\rho\sigma} \; .
      \end{align}
      \end{subequations}
      This shows compatibility of $\tilde{\nabla}$ with $[\tau, h]$,
      as claimed.  \qedhere
    \end{enumerate}
  \end{proof}
\end{prop}

\begin{rem}
  In the situation of \cref{prop:Gal_diff_dim_2}, a direct computation
  shows that the two connections $\nabla$ and $\tilde{\nabla}$ with
  difference tensor \eqref{eq:Gal_diff_dim_2} have agreeing
  Newton--Coriolis forms $\Omega = \tilde{\Omega}$ with respect to a
  choice $v$ of representative of the timelike direction.  Thus, we
  have shown that \emph{on a 2-dimensional manifold with a conformal
    Galilei structure $[\tau, h]$, when fixing the spatial torsion
    with respect to $[v]$ and the projective structure of a connection
    compatible with $[\tau, h]$, the remaining freedom is precisely
    the spatial part $P^\mu_\nu \varphi_\mu$ of the one-form $\varphi$
    encoding the connection's non-metricities}.
\end{rem}

\section{The Carroll case}
\label{sec:Carroll}

Here we develop a parallel story for the case of Carroll geometry, on
which see \textcite{Duval.EtAl:2014, Bergshoeff.EtAl:2014,
  Bergshoeff.EtAl:2017, March.Read:2025}.

\subsection{Prerequisites}

\begin{defn}
  Let $M$ be a differentiable manifold with $\dim M > 1$.
  \begin{enumerate}[label=(\roman*)]
  \item A \emph{Carroll structure} on $M$ is a pair $(v, \gamma)$ of a
    nowhere-vanishing vector field $v \in \Gamma(TM)$ and a
    positive-semidefinite $\gamma \in \Gamma(\bigvee^2 T^*M)$ of rank
    $n = \dim M - 1$ such that $v$ spans the degenerate direction of
    $\gamma$, i.e.
    \begin{equation}
      v^\mu \gamma_{\mu\nu} = 0.
    \end{equation}
  \item A \emph{conformal Carroll structure} on $M$ is an equivalence
    class $[v, \gamma]$ of Carroll structures under the equivalence
    relation $( v, \gamma) \sim ( \e^{- \frac{1}{2} \lambda } v,
    \e^\lambda \gamma )$ for $\lambda \in C^\infty (M)$.
  \item A linear connection $\nabla$ on $M$ is \emph{compatible} with
    a conformal Carroll structure if for some (and then any)
    representative $(v, \gamma)$ of the conformal structure there is a
    one-form $\varphi \in \Omega^1 (M)$ such that $\nabla v = -
    \frac{1}{2} \varphi \otimes v$, $\nabla \gamma = \varphi \otimes
    \gamma$.
  \end{enumerate}
\end{defn}

\begin{rem}
  \begin{enumerate}[label=(\roman*)]
  \item As in the Galilei case, we chose a definition of conformal
    equivalence that uses related factors for $v$ and $\gamma$, as
    this would result from the Carroll limit of a Lorentzian conformal
    transformation.  This is also the reason for the definition of
    compatibility using related one-forms.  This definition of
    conformal Carroll structures is also used, e.g., by
    \textcite{Herfray:2022}.

    As in the Galilei case, one could seek to relax this assumption.
  \item As in the pseudo-Riemannian and Galilei cases, the one-forms
    encoding the non-metricity of a compatible connection with respect
    to different representatives conformally related by $\lambda$ are
    related by $\varphi \mapsto \varphi + \D\lambda$.
  \item As in the Galilei case, the non-metricities of a connection
    with respect to a Carroll structure need to satisfy two
    conditions, see \textcite{Vigneron.Barzegar.Read:2025}.  The one
    relating only the non-metricities, namely $v^\mu \nabla_\rho
    \gamma_{\mu\nu} = - \gamma_{\mu\nu} \nabla_\rho v^\mu$, is
    satisfied for the non-metricities from the definition of
    compatibility.  The second condition involves the torsion: it
    reads
    \begin{subequations} \label{eq:Car_torsion}
    \begin{align}
      T_{(\mu\nu)\rho} v^\rho
      &= \frac{1}{2} v^\rho \nabla_\rho \gamma_{\mu\nu}
        + \gamma_{\rho(\mu} \nabla_{\nu)} v^\rho
        - \frac{1}{2} (\mathcal L_v \gamma)_{\mu\nu} \; , \\
      \intertext{where the first index on the torsion has been lowered
      with $\gamma$.  For the above non-metricities, this becomes}
      \label{eq:Car_conf_compat_torsion}
      T_{(\mu\nu)\rho} v^\rho
      &= \frac{1}{2} \varphi(v) \gamma_{\mu\nu}
        - \frac{1}{2} (\mathcal L_v \gamma)_{\mu\nu} \; .
    \end{align}
    \end{subequations}  
    Therefore, as in the Galilei case, we have to allow for torsionful
    connections in order to formulate a Carrollian analogue of Weyl's
    theorem.
  \end{enumerate}
\end{rem}

In order to fix the free parts of the torsion, we again need a notion
of timelike direction, or more specifically a co-direction in this
case:

\begin{defn}
  Given a conformal Carroll structure $[v, \gamma]$ on a
  differentiable manifold $M$, the notion of a \emph{timelike
    co-direction} (i.e.\ a timelike one-form up to rescaling) with
  respect to $[v, \gamma]$ is well-defined: it is an equivalence class
  $[\tau]$ of timelike one-forms, $\tau(v) \ne 0$, under the
  equivalence relation $\tau \sim \tilde{\lambda} \tau$ for
  $\tilde{\lambda} \in C^\infty(M, \R \setminus \{0\})$.  Given a
  timelike co-direction $[\tau]$, there is a well-defined unique
  \emph{spatial projector along $[\tau]$}, the vector bundle
  endomorphism $P$ of $T^*M$ defined by projecting onto $\ker v$
  (which is invariant under rescaling of $v$) and having kernel
  spanned by $\tau$ (which is invariant under rescaling of $\tau$).

  Fixing a representative $(v, \gamma)$ of the conformal structure,
  there is a unique representative $\tau$ of the timelike co-direction
  satisfying $\tau(v) = 1$, and in terms of these the spatial
  projector takes the usual explicit form $P^\mu_\nu = \delta^\mu_\nu
  - v^\mu \tau_\nu$.
\end{defn}

\begin{constr}
  \label{constr:Car_torsion_free}
  Let $[v, \gamma]$ be a conformal Carroll structure on $M$, and
  $[\tau]$ a timelike co-direction with respect to it.  Further, let
  $P$ be the spatial projector along $[\tau]$.  Fix some
  representative $(v, \gamma)$ of the conformal structure.

  By \eqref{eq:Car_torsion}, the components of a connection's torsion
  that are determined by the non-metricities are those of the form
  $T_{(\mu\nu)\rho} v^\rho$: we lower the first index with $\gamma$,
  symmetrise the first two indices, and take the temporal component on
  the third index (with respect to $(v, \gamma)$).  All other
  components of the torsion not determined by these may be freely
  specified in the definition of a connection (see
  \cref{sec:app_class_conn_Carroll} for a detailed discussion of the
  classification of connections on Carroll manifolds, based on
  \textcite{Vigneron.Barzegar.Read:2025}).

  The first set of freely specifiable torsion components is given by
  \begin{subequations}
  \begin{equation}
    T_{[\mu\nu]\rho} v^\rho \; ,
  \end{equation}
  where we still take the temporal component on the third index, but
  lower the first index with $\gamma$ and then \emph{anti}symmetrise
  the first two indices.  Together with the components
  $T_{(\mu\nu)\rho} v^\rho$ determined by the non-metricities, these
  are all torsion components of the form $T_{\mu\nu\rho} v^\rho$,
  i.e.\ spatial on the first index and mixed spatial--temporal on the
  second and third index.  The missing freely specifiable torsion
  components are thus the second set
  \begin{equation}
    P^\rho_\mu P^\sigma_\nu T_{\kappa\rho\sigma}
  \end{equation}
  which is spatial on all three indices, and the third set
  \begin{equation}
    (\delta^\rho_\sigma - P^\rho_\sigma) \tensor{T}{^\sigma_{\mu\nu}}
  \end{equation}
  \end{subequations}
  which is temporally projected (with respect to $[\tau]$) on the
  first index.

  Combined, we have thus seen that the freely specifiable torsion
  components with respect to a choice of timelike co-direction
  $[\tau]$ are
  \begin{equation}
    \label{eq:Car_torsion_free}
    \left(T_{[\mu\nu]\rho} v^\rho \; , \;
    P^\rho_\mu P^\sigma_\nu T_{\kappa\rho\sigma} \; , \;
    (\delta^\rho_\sigma - P^\rho_\sigma) \tensor{T}{^\sigma_{\mu\nu}}
    \right).
  \end{equation}
  This field may be understood as a section of the bundle
  \begin{equation}
    \textstyle
    \left(\bigwedge^2 \ker v\right)
    \oplus
    \left(\ker v \otimes \bigwedge^2 \ker v\right)
    \oplus
    \left(\Span\{v\} \otimes \bigwedge^2 T^*M\right).
  \end{equation}
  Note that the form of the free torsion components as given in
  \eqref{eq:Car_torsion_free} depends on the choice of representative
  $(v, \gamma)$ of the conformal Carroll structure.  For a discussion
  of the free torsion in a conformally invariant way, see
  \cref{sec:app_free_torsion_conf_Carroll}.
\end{constr}

\subsection{The Carroll Weyl-type theorem}

With these prerequisites, we can now prove a Weyl-type theorem, which
is however significantly weaker than in the case of Galilei geometry:

\begin{thm}
  \label{thm:Car_Weyl}
  Let $M$ be a differentiable manifold with $\dim M > 1$.  Let $[v,
  \gamma]$ be a conformal Carroll structure on $M$, and $[\tau]$ a
  timelike co-direction with respect to $[v, \gamma]$.

  \begin{enumerate}[label=(\roman*)]
  \item \label{thm:Car_Weyl_1} Let $\nabla, \tilde{\nabla}$ be two
    linear connections whose free torsion components with respect to
    $[\tau]$ agree, i.e.\ their torsions satisfy $T_{[\mu\nu]\rho}
    v^\rho = \tilde{T}_{[\mu\nu]\rho} v^\rho$, $P^\rho_\mu
    P^\sigma_\nu T_{\kappa\rho\sigma} = P^\rho_\mu P^\sigma_\nu
    \tilde{T}_{\kappa\rho\sigma}$, and $(\delta^\rho_\sigma -
    P^\rho_\sigma) \tensor{T}{^\sigma_{\mu\nu}} = (\delta^\rho_\sigma
    - P^\rho_\sigma) \tensor{\tilde{T}}{^\sigma_{\mu\nu}}$, where $P$
    is the spatial projector along $[\tau]$.\footnote{Note that the
      condition of agreeing free torsion is invariant under change of
      representative $(v, \gamma)$ of the conformal structure.}  If
    the connections are compatible with $[v, \gamma]$ and projectively
    equivalent, then their difference tensor $D := \nabla -
    \tilde{\nabla}$ is given by
    \begin{equation}
      \label{eq:Car_Weyl}
      D^\rho_{\mu\nu} = f \tau_\mu \delta^\rho_\nu
    \end{equation}
    for a function $f \in C^\infty(M)$ and a representative $\tau$ of
    the timelike co-direction.
  \item \label{thm:Car_Weyl_2} Conversely, let $\nabla$ be a linear
    connection that is compatible with $[v, \gamma]$.  Then for
    \emph{any} choice of $f$ and $\tau$, the connection
    $\tilde{\nabla} := \nabla - D$ with difference $D$ given by
    \eqref{eq:Car_Weyl} is compatible with $[v, \gamma]$ and has the
    same free torsion components with respect to $[\tau]$ as $\nabla$,
    and the connections are projectively equivalent.
  \end{enumerate}
\end{thm}

\begin{rem}
  \begin{enumerate}[label=(\roman*)]
  \item Note that even though the assumptions of part
    \ref{thm:Car_Weyl_1} of this theorem are the direct analogues of
    those in the Galilei case (\cref{thm:Gal_Weyl}), we cannot infer
    complete equality of the connections: we cannot prove covariant
    derivatives in temporal direction to agree.  Part
    \ref{thm:Car_Weyl_2} shows that this conclusion cannot be
    improved: any difference tensor of the form inferred in part
    \ref{thm:Car_Weyl_1} can indeed be realised.  Note also that this
    bears some similarity to the failure of the Galilean theorem for
    $\dim M = 2$ (\cref{prop:Gal_diff_dim_2}).
  \item While the connections $\nabla$ and $\tilde{\nabla}$ are not
    completely equal, derivatives with respect to them in spacelike
    directions, i.e.\ in the kernel of the chosen $[\tau]$, agree.  In
    particular, if $[\tau]$ is foliation forming, then both
    connections induce a unique spatial connection on the
    hypersurfaces integrating $\ker\tau$.  As identified in
    \textcite{Vigneron.Barzegar.Read:2025}, this spatial connection
    encodes precisely the surviving non-trivial piece of the Riemann
    tensor of a (torsion-free) connection that would be compatible
    with both a Galilei and a Carroll structure (i.e.,\ compatible
    with an Aristotelian structure---see e.g.\
    \textcite{Marotta.Szabo:2023}).  In this respect, because
    uniqueness holds in the Galilei case by \cref{thm:Gal_Weyl}, this
    partly elucidates why the only part that can differ between the
    two connections in the Carroll case is the derivative along $v$.
  \item Analogous to the Galilei case, for this theorem's statement we
    consider a \emph{fixed} choice of timelike co-direction $[\tau]$
    with respect to which we demand agreement of the free torsion
    components.
  \end{enumerate}
\end{rem}

\begin{proof}[Proof of \cref{thm:Car_Weyl}]
  Fix some representative $(v, \gamma)$ of the conformal structure,
  and let $\tau$ be the corresponding representative of the timelike
  co-direction satisfying $\tau(v) = 1$.

  \begin{enumerate}[label=(\roman*)]
  \item By projective equivalence, there is a one-form $\eta$ such
    that the difference tensor satisfies
    \begin{equation}
      \label{eq:Car_Weyl_pf_proj_eq}
      D^\rho_{(\mu\nu)}
      = \delta^\rho_{(\mu} \eta^{\vphantom{\rho}}_{\nu)} \; .
    \end{equation}
    Furthermore, both connections having the same free torsion
    components with respect to $[\tau]$ implies in particular
    \begin{subequations}
    \begin{align}
      P^\rho_\mu P^\sigma_\nu \gamma_{\kappa\lambda}
        D^\lambda_{[\rho\sigma]}
      &= 0, \\
      \shortintertext{implying}
      \label{eq:Car_Weyl_pf_diff_proj}
      P^\rho_\mu P^\sigma_\nu \gamma_{\kappa\lambda}
        D^\lambda_{\rho\sigma}
      &= P^\rho_\mu P^\sigma_\nu \gamma_{\kappa\lambda}
        D^\lambda_{(\rho\sigma)}
        = P^\rho_\mu P^\sigma_\nu \gamma_{\kappa(\rho} \eta_{\sigma)}
        \; .
    \end{align}
    \end{subequations}
    It also implies
    \begin{subequations}
    \begin{align}
      \tau_\rho D^\rho_{[\mu\nu]}
      &= 0, \\
      \shortintertext{implying}
      \label{eq:Car_Weyl_pf_diff_tau}
      \tau_\rho D^\rho_{\mu\nu}
      &= \tau_\rho D^\rho_{(\mu\nu)}
        = \tau_{(\mu} \eta_{\nu)} \; .
    \end{align}
    \end{subequations}

    By compatibility of the connections with $[v, \gamma]$, there are
    one-forms $\varphi, \tilde{\varphi}$ such that
    \begin{equation}
      \nabla v = -\frac{1}{2} \varphi \otimes v,
      \tilde{\nabla} v = -\frac{1}{2} \tilde{\varphi} \otimes v,
      \nabla \gamma = \varphi \otimes \gamma,
      \tilde{\nabla} \gamma = \tilde{\varphi} \otimes \gamma.
    \end{equation}

    Combining the two equations for compatibility with $v$, we obtain
    \begin{align}
      \label{eq:Car_Weyl_pf_compat_v}
      -\frac{1}{2} \varphi_\mu v^\rho
      &= \nabla_\mu v^\rho \nonumber\\
      &= \tilde{\nabla}_\mu v^\rho + D^\rho_{\mu\nu} v^\nu \nonumber\\
      &= -\frac{1}{2} \tilde{\varphi}_\mu v^\rho
        + D^\rho_{\mu\nu} v^\nu .
    \end{align}
    From this, we will now derive three equations.  First, contracting
    with $\gamma_{\rho\sigma}$, we obtain
    \begin{equation}
      \label{eq:Car_Weyl_pf_diff_gamma_v}
      \gamma_{\rho\sigma} D^\rho_{\mu\nu} v^\nu = 0.
    \end{equation}

    Second, contracting \eqref{eq:Car_Weyl_pf_compat_v} with $v^\mu$,
    we obtain
    \begin{align}
      -\frac{1}{2} \varphi(v) v^\rho
      &= -\frac{1}{2} \tilde{\varphi}(v) v^\rho
        + D^\rho_{(\mu\nu)} v^\mu v^\nu \nonumber\\
      &= -\frac{1}{2} \tilde{\varphi}(v) v^\rho + \eta(v) v^\rho , \\
      \intertext{where we have used \eqref{eq:Car_Weyl_pf_proj_eq}.
      This directly implies}
      \label{eq:Car_Weyl_pf_1}
      -\frac{1}{2} \varphi(v)
      &= -\frac{1}{2} \tilde{\varphi}(v) + \eta(v).
    \end{align}

    Third, contracting \eqref{eq:Car_Weyl_pf_compat_v} with $2
    \tau_\rho$ yields
    \begin{equation}
      -\varphi_\mu
      = -\tilde{\varphi}_\mu + 2 \tau_{(\mu} \eta_{\nu)} v^\nu
      = -\tilde{\varphi}_\mu + \eta(v) \tau_\mu + \eta_\mu \; ,
    \end{equation}
    where we used \eqref{eq:Car_Weyl_pf_diff_tau}.  Contracting this
    with $h^{\mu\nu}$---the components of the contravariant space
    metric with respect to $\tau$, defined by the requirements
    $h^{\mu\nu} = h^{\nu\mu}$, $h^{\mu\nu} \tau_\nu = 0$, and
    $h^{\mu\nu} \gamma_{\nu\rho} = P^\mu_\rho$---yields
    \begin{equation}
      \label{eq:Car_Weyl_pf_2}
      -\varphi^\nu = -\tilde{\varphi}^\nu + \eta^\nu ,
    \end{equation}
    where $\varphi^\nu := h^{\mu\nu} \varphi_\mu$ is the
    $\tau$-spacelike vector field corresponding to $\varphi$, and
    analogously for $\tilde{\varphi}$ and $\eta$.

    On the other hand, combining the two equations for compatibility
    with $\gamma$, we obtain
    \begin{align}
      \label{eq:Car_Weyl_pf_compat_gamma_1}
      \varphi_\rho \gamma_{\mu\nu}
      &= \nabla_\rho \gamma_{\mu\nu} \nonumber\\
      &= \tilde{\nabla}_\rho \gamma_{\mu\nu}
        - D^\kappa_{\rho\mu} \gamma_{\kappa\nu}
        - D^\kappa_{\rho\nu} \gamma_{\mu\kappa}
        \nonumber\\
      &= \tilde{\varphi}_\rho \gamma_{\mu\nu}
        - 2 D^\kappa_{\rho(\mu}
          \gamma^{\vphantom{\kappa}}_{\nu)\kappa} \; .
    \end{align}
    Contracting this with $h^{\rho\sigma}$ yields
    \begin{equation}
      \label{eq:Car_Weyl_pf_compat_gamma_2}
      \varphi^\sigma \gamma_{\mu\nu}
      = \tilde{\varphi}^\sigma \gamma_{\mu\nu}
        - 2 h^{\rho\sigma} D^\kappa_{\rho(\mu}
          \gamma^{\vphantom{\kappa}}_{\nu)\kappa} \; .
    \end{equation}
    In order to simplify \eqref{eq:Car_Weyl_pf_compat_gamma_2}, we
    consider its last term.  From \eqref{eq:Car_Weyl_pf_diff_gamma_v},
    we know that the expression $\gamma_{\rho\sigma} D^\rho_{\mu\nu}$
    vanishes when contracted with $v^\nu$; i.e.\ the index $\nu$ is
    spacelike.  Therefore, the expression is equal to its projection
    onto space along $[\tau]$ on this index,
    \begin{subequations}
    \begin{align}
      \gamma_{\rho\sigma} D^\rho_{\mu\nu}
      &= \gamma_{\rho\sigma} P^\lambda_\nu D^\rho_{\mu\lambda} \; . \\
      \intertext{Contracting with $h^{\mu\kappa}$, we have}
      \gamma_{\rho\sigma} h^{\mu\kappa} D^\rho_{\mu\nu}
      &= \gamma_{\rho\sigma} h^{\mu\kappa} P^\lambda_\nu
        D^\rho_{\mu\lambda} \nonumber\\
      &= h^{\kappa\alpha} \gamma_{\rho\sigma} P^\mu_\alpha
        P^\lambda_\nu D^\rho_{\mu\lambda} \; . \\
      \intertext{But the totally spacelike projected difference tensor can be
      expressed in terms of $\eta$ according to
      \eqref{eq:Car_Weyl_pf_diff_proj},
      giving}
      \label{eq:Car_Weyl_pf_diff_spatial}
      \gamma_{\rho\sigma} h^{\mu\kappa} D^\rho_{\mu\nu}
      &= h^{\kappa\alpha} P^\mu_\alpha P^\lambda_\nu
        \gamma_{\sigma(\mu} \eta_{\lambda)} \nonumber\\
      &= h^{\kappa\mu} P^\lambda_\nu \gamma_{\sigma(\mu}
        \eta_{\lambda)} \nonumber\\
      &= \frac{1}{2} h^{\kappa\mu} P^\lambda_\nu
        (\gamma_{\sigma\mu} \eta_\lambda
          + \gamma_{\sigma\lambda} \eta_\mu) \nonumber\\
      &= \frac{1}{2}
        (P^\kappa_\sigma P^\lambda_\nu \eta_\lambda
          + \gamma_{\sigma\nu} \eta^\kappa).
    \end{align}
    \end{subequations}
    Now using \eqref{eq:Car_Weyl_pf_diff_spatial}, the compatibility
    equation \eqref{eq:Car_Weyl_pf_compat_gamma_2} takes the form
    \begin{equation}
      \label{eq:Car_Weyl_pf_compat_gamma_3}
      \varphi^\sigma \gamma_{\mu\nu}
      = \tilde{\varphi}^\sigma \gamma_{\mu\nu}
        - P^\sigma_{(\nu} P^\lambda_{\mu)} \eta_\lambda
        - \gamma_{\mu\nu} \eta^\sigma .
    \end{equation}
    Contracting this with $h^{\mu\nu}$ yields
    \begin{equation}
      \label{eq:Car_Weyl_pf_3}
      n \varphi^\sigma
      = n \tilde{\varphi}^\sigma - \eta^\sigma - n \eta^\sigma
      = n \tilde{\varphi}^\sigma - (n+1) \eta^\sigma ,
    \end{equation}
    where $n = \dim M - 1$.  Adding $n \cdot \eqref{eq:Car_Weyl_pf_2}$
    to \eqref{eq:Car_Weyl_pf_3} (and renaming the free index), we obtain
    \begin{subequations}
    \begin{align}
      0 &= - \eta^\rho , \\
      \shortintertext{i.e.}
      \eta &= \eta(v) \tau.
    \end{align}
    \end{subequations}

    Inserting this result, using \eqref{eq:Car_Weyl_pf_proj_eq} we may
    write the difference tensor as
    \begin{equation}
      \label{eq:Car_Weyl_pf_diff}
      D^\rho_{\mu\nu}
      = D^\rho_{(\mu\nu)} + D^\rho_{[\mu\nu]}
      = \delta^\rho_{(\mu} \eta^{\vphantom{\rho}}_{\nu)}
        + \frac{1}{2} \tensor{(\Delta T)}{^\rho_{\mu\nu}}
      = \eta(v) \delta^\rho_{(\mu} \tau^{\vphantom{\rho}}_{\nu)}
        + \frac{1}{2} \tensor{(\Delta T)}{^\rho_{\mu\nu}} \; ,
    \end{equation}
    where $\Delta T := T - \tilde{T}$ denotes the difference between
    the two connections' torsions.  By performing a
    spacelike--timelike decomposition (i.e.\ by using $\delta^\mu_\nu
    = P^\mu_\nu + v^\mu \tau_\nu$), a straightforward calculation
    shows that the torsions may be expressed in terms of their free
    components with respect to $[\tau]$ and their constrained
    components as
    \begin{equation}
      \tensor{T}{^\rho_{\mu\nu}}
      = (\delta^\rho_\sigma - P^\rho_\sigma)
          \tensor{T}{^\sigma_{\mu\nu}}
        + \big(h^{\rho\sigma} T_{(\sigma\mu)\kappa} v^\kappa \tau_\nu
          + h^{\rho\sigma} T_{[\sigma\mu]\kappa} v^\kappa \tau_\nu
          - (\mu \leftrightarrow \nu)\big)
        + h^{\rho\sigma} P^\kappa_\nu P^\lambda_\mu
          T_{\sigma\lambda\kappa} \; .
    \end{equation}
    Since we assumed the free torsion components with respect to
    $[\tau]$ to agree, the torsion difference is thus given by
    \begin{equation}
      \label{eq:Car_Weyl_pf_torsion_diff}
      \tensor{(\Delta T)}{^\rho_{\mu\nu}}
      = h^{\rho\sigma} (\Delta T)_{(\sigma\mu)\kappa}
          v^\kappa \tau_\nu
        - (\mu \leftrightarrow \nu).
    \end{equation}
    Now by \eqref{eq:Car_conf_compat_torsion} the constrained torsion
    components are given by
    \begin{align}
      \accentset{(\sim)}{T}_{(\sigma\mu)\kappa} v^\kappa
      &= \frac{1}{2} \accentset{(\sim)}{\varphi}(v) \gamma_{\sigma\mu}
        - \frac{1}{2} (\mathcal L_v \gamma)_{\sigma\mu} \; . \\
      \shortintertext{Hence their difference is}
      (\Delta T)_{(\sigma\mu)\kappa} v^\kappa
      &= \frac{1}{2} (\varphi(v) - \tilde{\varphi}(v)) \gamma_{\sigma\mu}
        \nonumber\\
      &= -\eta(v) \gamma_{\sigma\mu} \; ,
    \end{align}
    where we used \eqref{eq:Car_Weyl_pf_1}.  Inserting this into
    \eqref{eq:Car_Weyl_pf_torsion_diff}, the torsion difference
    evaluates to
    \begin{align}
      \tensor{(\Delta T)}{^\rho_{\mu\nu}}
      &= -h^{\rho\sigma} \eta(v) \gamma_{\sigma\mu} \tau_\nu
        - (\mu \leftrightarrow \nu) \nonumber\\
      &= -\eta(v) P^\rho_\mu \tau_\nu
        - (\mu \leftrightarrow \nu) \nonumber\\
      &= 2 \eta(v) P^\rho_{[\nu} \tau^{\vphantom{\rho}}_{\mu]}
        \nonumber\\
      &= 2 \eta(v) (\delta^\rho_{[\nu}
          - v^\rho \tau^{\vphantom{\rho}}_{[\nu})
        \tau^{\vphantom{\rho}}_{\mu]} \nonumber\\
      &= 2 \eta(v) \delta^\rho_{[\nu} \tau^{\vphantom{\rho}}_{\mu]} \; .
    \end{align}
    Inserting this into \eqref{eq:Car_Weyl_pf_diff}, the difference
    tensor is given by
    \begin{equation}
      D^\rho_{\mu\nu}
      = \eta(v) \delta^\rho_{(\mu} \tau^{\vphantom{\rho}}_{\nu)}
        + \eta(v) \delta^\rho_{[\nu} \tau^{\vphantom{\rho}}_{\mu]}
      = \eta(v) \delta^\rho_\nu \tau_\mu \; .
    \end{equation}
    Thus, we have proved that it takes the form from the theorem
    statement with $f = \eta(v)$.

  \item Since $\nabla$ is compatible with $[v, \gamma]$, there is a
    one-form $\varphi$ such that
    \begin{equation}
      \nabla v = -\frac{1}{2} \varphi \otimes v,
      \nabla \gamma = \varphi \otimes \gamma.
    \end{equation}
    Let $f \in C^\infty(M)$ be an arbitrary smooth function, and
    define $\tilde{\nabla} := \nabla - D$ for
    \begin{equation}
      D^\rho_{\mu\nu} = f \tau_\mu \delta^\rho_\nu \; .
    \end{equation}

    We have to show that $\tilde{\nabla}$ is compatible with $[v,
    \gamma]$, that its free torsion components with respect to
    $[\tau]$ agree with those of $\nabla$, and that the two
    connections are projectively equivalent.  All of this follows by
    direct computations, which we will spell out in the following.

    First, considering projective equivalence, by definition we
    directly obtain
    \begin{equation}
      D^\rho_{(\mu\nu)}
      = f \tau^{\vphantom{\rho}}_{(\mu} \delta^\rho_{\nu)}
      = \delta^\rho_{(\mu} \eta^{\vphantom{\rho}}_{\nu)}
    \end{equation}
    with $\eta := f \tau$, showing that $\nabla$ and $\tilde{\nabla}$
    are projectively equivalent.

    Next, we turn to compatibility with the conformal structure.  We
    compute
    \begin{align}
      \tilde{\nabla}_\mu v^\rho
      &= \nabla_\mu v^\rho - D^\rho_{\mu\nu} v^\nu \nonumber\\\displaybreak[0]
      &= -\frac{1}{2} \varphi_\mu v^\rho
        - f \tau_\mu \delta^\rho_\nu v^\nu \nonumber\\
      &= -\frac{1}{2} (\varphi_\mu + 2 f \tau_\mu) v^\rho
    \end{align}
    and
    \begin{align}
      \tilde{\nabla}_\rho \gamma_{\mu\nu}
      &= \nabla_\rho \gamma_{\mu\nu} 
        + D^\kappa_{\rho\mu} \gamma_{\kappa\nu}
        + D^\kappa_{\rho\nu} \gamma_{\mu\kappa} \nonumber\\
      &= \varphi_\rho \gamma_{\mu\nu}
        + f \tau_\rho \delta^\kappa_\mu \gamma_{\kappa\nu}
        + f \tau_\rho \delta^\kappa_\nu \gamma_{\mu\kappa} \nonumber\\
      &= (\varphi_\rho + 2 f \tau_\rho) \gamma_{\mu\nu} \; .
    \end{align}
    This means
    \begin{equation}
      \tilde{\nabla} v = -\frac{1}{2} \tilde{\varphi} \otimes v,
      \tilde{\nabla} \gamma = \tilde{\varphi} \otimes \gamma
    \end{equation}
    with $\tilde{\varphi} = \varphi + 2 f \tau$, showing that
    $\tilde{\nabla}$ is compatible with $[v, \gamma]$.

    Finally, the difference $\Delta T := T - \tilde{T}$ between the
    connections' torsions is given by
    \begin{align}
      \tensor{(\Delta T)}{^\rho_{\mu\nu}}
      &= 2 D^\rho_{[\mu\nu]}
        = 2 f \tau^{\vphantom{\rho}}_{[\mu} \delta^\rho_{\nu]} \; . \\
      \intertext{This directly implies}
      \label{eq:Car_Weyl_pf_diff_free_tor_1}
      (\delta^\rho_\sigma - P^\rho_\sigma)
        \tensor{(\Delta T)}{^\sigma_{\mu\nu}}
      &= v^\rho \tau_\sigma 2 f \tau^{\vphantom{\sigma}}_{[\mu}
          \delta^\sigma_{\nu]}
        = 2 f v^\rho \tau_{[\mu} \tau_{\nu]}
        = 0. \\
      \intertext{Further, we have}
      (\Delta T)_{\sigma\mu\nu}
      &= \gamma_{\sigma\rho} \tensor{(\Delta T)}{^\rho_{\mu\nu}}
        \nonumber\\
      &= \gamma_{\sigma\rho} 2 f \tau^{\vphantom{\rho}}_{[\mu}
        \delta^\rho_{\nu]} \nonumber\\
      &= 2 f \tau_{[\mu} \gamma_{\nu]\sigma} \nonumber\\
      &= f \tau_\mu \gamma_{\nu\sigma} - f \tau_\nu \gamma_{\mu\sigma}
        \; , \\
      \intertext{from which we obtain}
      \label{eq:Car_Weyl_pf_diff_free_tor_2}
      (\Delta T)_{[\mu\nu]\rho} v^\rho
      &= (f \tau_{[\nu} \gamma_{|\rho|\mu]}
        - f \tau_\rho \gamma_{[\nu\mu]}) v^\rho
      = 0 \\
      \intertext{as well as}
      \label{eq:Car_Weyl_pf_diff_free_tor_3}
      P^\mu_\lambda P^\nu_\kappa (\Delta T)_{\sigma\mu\nu}
      &= P^\mu_\lambda P^\nu_\kappa (f \tau_\mu \gamma_{\nu\sigma}
        - f \tau_\nu \gamma_{\mu\sigma})
      = 0.
    \end{align}
    But \eqref{eq:Car_Weyl_pf_diff_free_tor_1},
    \eqref{eq:Car_Weyl_pf_diff_free_tor_2}, and
    \eqref{eq:Car_Weyl_pf_diff_free_tor_3} mean that
    $(\delta^\rho_\sigma - P^\rho_\sigma) \tensor{T}{^\sigma_{\mu\nu}}
    = (\delta^\rho_\sigma - P^\rho_\sigma)
    \tensor{\tilde{T}}{^\sigma_{\mu\nu}}$, $T_{[\mu\nu]\rho} v^\rho =
    \tilde{T}_{[\mu\nu]\rho} v^\rho$, and $P^\mu_\lambda P^\nu_\kappa
    T_{\sigma\mu\nu} = P^\mu_\lambda P^\nu_\kappa
    \tilde{T}_{\sigma\mu\nu}$---i.e.\ the connections' free torsion
    components with respect to $[\tau]$ agree.  \qedhere
  \end{enumerate}
\end{proof}

\subsection{A Carroll analogue of Weyl metrics}

We may define an analogue of Weyl metrics for Carroll geometry as
follows, in parallel to the Galilei case.

\begin{constr}
  Fixing a Carroll structure $(v, \gamma)$ on $M$, with respect to a
  choice of unit timelike one-form $\tau \in \Omega^1(M)$, $\tau(v) =
  1$, any linear connection $\nabla$ on $M$ is uniquely determined
  by---and uniquely determines, vice versa---its $v$-non-metricity
  $\nabla v$, the spatial part $P^\rho_\mu P^\sigma_\nu \nabla_\kappa
  \gamma_{\rho\sigma}$ of its $\gamma$-non-metricity, its free torsion
  components $T_\text{free}$ with respect to $[\tau]$ as discussed in
  \cref{constr:Car_torsion_free}, and the field $\hat\chi_{\mu\nu} =
  P^\rho_\mu P^\sigma_\nu \nabla_{(\rho} \tau_{\sigma)}$; see
  \textcite{Vigneron.Barzegar.Read:2025} and our discussion in
  \cref{sec:app_class_conn_Carroll}.  Assuming the non-metricities to
  be of the form $\nabla v = - \frac{1}{2} \varphi \otimes v$, $\nabla
  \gamma = \varphi \otimes \gamma$, such that $\nabla$ is compatible
  with the conformal Carroll structure $[v, \gamma]$, we already know
  that fixing the connection $\nabla$, under a conformal rescaling
  $(v,\gamma) \mapsto (\e^{-\frac{1}{2}\lambda} v, \e^\lambda \gamma)$
  with $\lambda \in C^\infty(M)$ the one-form $\varphi$ transforms as
  $\varphi \mapsto \varphi + \D\lambda$.  Further, the free torsion
  may be written in a conformally invariant way, see
  \cref{sec:app_free_torsion_conf_Carroll}.  Rescaling $\tau \mapsto
  \e^{\frac{1}{2}\lambda}$ such that it remains unit timelike with
  respect to the new representative of the conformal structure, a
  direct computation shows that the $\hat\chi$ field of $\nabla$ with
  respect to $\tau$ transforms as $\hat\chi \mapsto
  \e^{\frac{1}{2}\lambda} \hat\chi$.

  Therefore, a conformal Carroll structure $[v, \gamma]$ on $M$
  together with a compatible linear connection is equivalently
  characterised by a triple $([v, \gamma, \varphi, \hat\chi], [\tau],
  T_\text{free})$ consisting of (i)~an equivalence class of Carroll
  structures $(v, \gamma)$, one-forms $\varphi$, and symmetric tensor
  fields $\hat\chi \in \bigvee^2 \ker v$ on $M$ under the equivalence
  relation
  \begin{equation}
    (v, \gamma, \varphi, \hat\chi)
    \sim (\e^{-\frac{1}{2}\lambda} v, \e^\lambda \gamma,
      \varphi + \D\lambda, \e^{\frac{1}{2}\lambda} \hat\chi)
  \end{equation}
  for $\lambda \in C^\infty(M)$; (ii)~a timelike co-direction $[\tau]$
  with respect to the conformal Carroll structure $[v, \gamma]$; and
  (iii)~a section $T_\text{free}$ of the bundle in which the
  conformally invariant free torsion takes its values, as discussed in
  \cref{constr:Car_torsion_free_space}, in particular
  \eqref{eq:Car_torsion_free_space}.

  In order to reformulate our Weyl-type theorem for Carroll geometry
  (\cref{thm:Car_Weyl}) in terms of such a `Carroll Weyl metric' $([v,
  \gamma, \varphi, \hat\chi], [\tau], T_\text{free})$, we have to
  examine the influence of the non-uniqueness in the theorem statement
  on the fields defining the connection.  A direct computation shows
  that, fixing $[v, \gamma]$ and $[\tau]$, a change of the connection
  according to $\nabla \mapsto \nabla - D$ as allowed by the
  theorem---i.e.\ with $D^\rho_{\mu\nu} = f \tau_\mu
  \delta^\rho_\nu$---leaves $\hat\chi$ invariant.  In the proof of
  \cref{thm:Car_Weyl} we have further seen that and how $\varphi$
  changes---namely by addition of $2f\tau$.

  Thus, we may reformulate \cref{thm:Car_Weyl} as the statement that
  \emph{a `Carroll Weyl metric' $([v, \gamma, \varphi, \hat\chi],
    [\tau], T_\text{\emph{free}})$ is uniquely determined by its
    conformal and projective structures together with its timelike
    co-direction $[\tau]$ and its free torsion $T_\text{\emph{free}}$
    with respect to $[\tau]$, up to addition of (functional) multiples
    of $\tau$ to $\varphi$}.
\end{constr}

\begin{rem}
  One strategy which could be implemented in order to save some form
  of uniqueness in \cref{thm:Car_Weyl} would be to enlarge the
  definition of a Carroll Weyl metric to include equivalence under the
  transformation $\varphi \to \varphi + f \tau$.  However, the
  corresponding \emph{connection} would still depend on the chosen
  co-direction $[\tau]$, contrary to the Galilei case for which it
  does not depend on the initially chosen direction $[v]$.  Therefore,
  this route to save uniqueness would still suffer from more severe
  restrictions as compared with the Galilei case.  In the conclusion,
  we comment on another possible route, involving the `ambient
  perspective' of \textcite{Jia.Karydas.Leigh:2023}.
\end{rem}

\section{Conclusions}
\label{sec:concl}

In this article, we have proved versions of one famous theorem of
\textcite{Weyl:1921}---that a Weyl metric is fixed uniquely by its
associated projective and conformal structures---for the cases of
Galilei and Carroll geometries; for both of the latter two cases, we
used the appropriate notion of conformal structure as would arise in
the relevant `non-relativistic' or `ultra-relativistic' limit.  Our
Galilean result generalises and extends prior work by
\textcite{Ewen.Schmidt:1989, Curiel:2015, March:2025}; our Carrollian
result is entirely new to the literature---although also interestingly
weaker than those for the pseudo-Riemannian and Galilei cases, as the
equivalent uniqueness result does \emph{not} hold in the Carroll
context.  A concise summary of our results may be found in
\cref{tab:summary}.

\newcommand{\vspacetable}[1]{\multicolumn{2}{c}{} \vspace{#1}\\}
\afterpage{
\newgeometry{top=1cm,bottom=2.25cm,right=2.44cm,left=2.44cm}

\begin{landscape}
\begin{table}
  \centering
  \captionsetup{width=20cm}
  \caption{\label{tab:summary}%
    Summary of definitions and properties for conformal and projective
    structures in the pseudo-Riemannian, Galilean, and Carrollian
    cases.  We let $\lambda, f \in C^\infty(M)$, $\tilde{\lambda} \in
    C^\infty(M, \R \setminus \{0\})$, and $\eta, \varphi \in
    \Omega^1(M)$.}

  \renewcommand{\arraystretch}{2.2}
  \renewcommand{\tabcolsep}{10pt}
  \begin{tabular}{llll}
  \toprule

  \vspacetable{-40pt}
  \centering \textbf{Quantities}
  & \textbf{pseudo-Riemannian geometry}
  & \textbf{Galilei geometry}
  & \textbf{Carroll geometry} \\
  \midrule

  \vspacetable{-40pt}
  \multicolumn{4}{c}{\textls{\textsc{Definitions}}} \\
  \vspacetable{-40pt}

  Projective structure $\mathcal P = [\nabla]$
  & $\nabla \sim \nabla + D$ with $D^\rho_{(\mu\nu)} =
    \delta^\rho_{(\mu} \eta^{\vphantom{\rho}}_{\nu)}$
  & same
  & same \\

  Metric structure
  & $g$
  & $(\tau, h)$ with $\tau_\mu h^{\mu\nu} = 0$
  & $(v, \gamma)$ with $v^\mu \gamma_{\mu\nu} = 0$ \\

  Conformal structure
  & \makecell[l]{$[g]$ with \\
    \quad $g \sim \e^\lambda g$}
  & \makecell[l]{$[\tau,h]$ with \\
    \quad $(\tau, h)
      \sim (\e^{\frac{1}{2}\lambda} \tau, \e^{-\lambda} h)$}
  & \makecell[l]{$[v,\gamma]$ with \\
    \quad $(v, \gamma)
      \sim (\e^{-\frac{1}{2}\lambda} v, \e^\lambda \gamma)$}
  \\
  \vspacetable{-35pt}

  \makecell[l]{Conformally compatible\\ connection}
  & $\nabla g = \varphi \otimes g$
  & \makecell[l]{$\nabla \tau = \tfrac{1}{2} \varphi \otimes \tau$,
    \quad $\nabla h = - \varphi \otimes h$}
  & \makecell[l]{$\nabla v = -\tfrac{1}{2} \varphi \otimes v$,
    \quad $\nabla \gamma = \varphi \otimes \gamma$}
  \\
  \vspacetable{-35pt}

  Timelike (co-)direction
  & --
  & $[v]$ with $\tau(v) \ne 0$, $v \sim \tilde{\lambda} v$
  & $[\tau]$ with $\tau(v) \ne 0$, $\tau \sim \tilde{\lambda} \tau$
  \\
  \vspacetable{-30pt}

  \makecell[l]{Spatial projector\\ \hspace{0em}}
  & \makecell[l]{--\\ \hspace{0em}}
  & \makecell[l]{$P^\mu_\nu = \delta^\mu_\nu - \hat v^\mu \tau_\nu$
    where $\hat v \in [v]$ \\ \quad such that $\tau(\hat v) = 1$}
  & \makecell[l]{$P^\mu_\nu = \delta^\mu_\nu - v^\mu \hat\tau_\nu$
    where $\hat\tau \in [\tau]$ \\ \quad such that $\hat\tau(v) = 1$}
  \\
  \vspacetable{-30pt}

  \makecell[l]{Free parts of the torsion ($T_\text{free}$) \\
    given fixed non-metricities}
  & \makecell[l]{ $\tensor{T}{^\rho_{\mu\nu}}$ \\
    \hspace{0em}}
  & \makecell[l]{$P^\rho_\sigma \tensor{T}{^\rho_{\mu\nu}}$ \\
    \quad with respect to a choice of $[v]$}
  & \makecell[l]{$\gamma^{\vphantom{\kappa}}_{\kappa[\mu}
      \tensor{T}{^\kappa_{\nu]\rho}} v^\rho$,
    $P^\rho_\mu P^\sigma_\nu \gamma_{\kappa\alpha}
    \tensor{T}{^\alpha_{\rho\sigma}}$,
    $(\delta^\rho_\sigma - P^\rho_\sigma)
      \tensor{T}{^\sigma_{\mu\nu}}$ \\
    \quad with respect to a choice of $[\tau]$}
  \\
  \vspacetable{-30pt}

  \midrule

  \vspacetable{-40pt}
  \multicolumn{4}{c}{\textls{\textsc{Properties}}}\\
  \vspacetable{-30pt}

  \makecell[l]{Weyl-type theorem \\ \hspace{0em} \\ \hspace{0em} \\
    \hspace{0em}}
  & \makecell[l]{\Cref{thm:pseudo-Riemann_Weyl_torsion}: For $\dim M >
      1$,\\
    given $[g]$ and $\mathcal P$ with fixed $T$,\\
    $\nabla \in \mathcal P$ compatible with $[g]$\\
    is unique.}
  & \makecell[l]{\Cref{thm:Gal_Weyl}: For $\dim M > 2$,\\
    given $[\tau, h]$ and $\mathcal P$ with fixed $T_\text{free}$,\\
    $\nabla \in \mathcal P$ compatible with $[\tau, h]$\\
    is unique.}
  & \makecell[l]{\Cref{thm:Car_Weyl}: \\
    Given $[v, \gamma]$ and $\mathcal P$ with fixed $T_\text{free}$,\\
    $\nabla \in \mathcal P$ compatible with $[v, \gamma]$\\
    is unique up to $D^\rho_{\mu\nu} = f \tau_\mu \delta^\rho_\nu$.}
  \\
  \vspacetable{-20pt}

  \makecell[l]{Connection and \\ (torsional) Weyl metric \\
    \hspace{0em}} 
  & \makecell[l]{$([g], \nabla) \longleftrightarrow
    ([g,\varphi], T_\text{free})$ \\
    \quad with $(g, \varphi) \sim (\e^\lambda g, \varphi + \D\lambda)$ \\
    \hspace{0em}}
  & \makecell[l]{$([\tau, h], \nabla) \longleftrightarrow
    ([\tau, h, \varphi, \Omega], [v], T_\text{free})$ \\
    \quad with $(\tau, h, \varphi, \Omega)$ \\
      \qquad $\sim (\e^{\frac{1}{2}\lambda} \tau, \e^{-\lambda} h,
        \varphi + \D\lambda, \e^{-\frac{1}{2}\lambda} \Omega)$}
  & \makecell[l]{$([v, \gamma], \nabla) \longleftrightarrow
    ([v, \gamma, \varphi, \hat\chi], [\tau], T_\text{free})$ \\
    \quad with $(v, \gamma, \varphi, \hat\chi)$ \\
      \qquad $\sim (\e^{-\frac{1}{2}\lambda} v, \e^\lambda \gamma,
        \varphi + \D\lambda, \e^{\frac{1}{2}\lambda} \hat\chi)$}
  \\

  \bottomrule
  \end{tabular}
\end{table}
\end{landscape}

\aftergroup\restoregeometry
}

There are various payoffs of this work.  From a conceptual point of
view, one merit of undertaking this work is that one comes to
understand better the `sub-metrical' constituents of both Galilei and
Carroll geometries.  This is particularly pertinent since in recent
years there has been heightened interest in conformal structures in
both the Galilei and Carroll contexts.  One motivation for this
interest has been due to the application of these non-relativistic
conformal geometries in the context of holography; the work has mostly
so far had to do with Galilean conformal geometry---see e.g.\
\textcite{Bagchi.Gopakumar:2009, Bagchi.Mandal:2009}.  Another
motivation has had to do with the study of non-relativistic field
theories more generally---see e.g.\ \textcite{Hagen:1972,
  Alishahiha.Davody.Vahedi:2009, Bagchi.Chakrabortty.Mehra:2018,
  Chen.Liu:2023, Bagchi.Basu.Mehra:2014, Duval.Horvathy:2011} for
conformal structures in Galilean field theory, and e.g.\
\textcite{Gupta.Suryanarayana:2021, Chen.Sun.Zheng:2024, Dutta:2024,
  Bergshoeff.EtAl:2026, Duval.EtAl:2014b} for conformal structures in
Carrollian field theory.  A further motivation has to do with the
study of asymptotic symmetries and the BMS group; to these issues it
is Carroll conformal structures which are most pertinent---see e.g.\
\textcite{Ciambelli.EtAl:2019}.  Finally, conformal Galilei structures
have found application in work on non-relativistic twistor
theory---see \textcite{Dunajski.Gundry:2016, Dunajski.Penrose:2023,
  March.Read:2026}---which raises the question of whether conformal Carroll
structures of the kind discussed in this article could find
application to ultra-relativistic versions of twistor theory.

One might, indeed, anticipate that the results presented in this
article could interact quite directly with the above contexts.  For
example, if one is in the context of working with a conformal Galilei
or Carroll structure (perhaps e.g.\ working with holography or
asymptotics), then the introduction of a preferred class of observers
would introduce a proper part of the projective structure, which would
in turn go some way to deriving a Galilei/Carroll Weyl spacetime.
(Here, `preferred class of observers' is to be interpreted as not
meaning a single congruence, but rather a choice of worldline for each
timelike direction at any point.  Note, moreover, that this
underdetermines the \emph{spacelike} unparametrised geodesics, which
would have to be fixed in some other way; the physics of how this
might be fixed will vary depending upon context.)  As such, this work
might well serve as a simple illustration of the `emergence of
spacetime' which has aroused physicists' interests in (for example)
the---admittedly much more advanced!---context of `dressing to the
observer' in algebraic quantum field theory.\footnote{See e.g.\
  \textcite{Witten:2022, Chandrasekaran.EtAl:2023}.}

Another potentially fruitful avenue for further pursuit would be to
consider the `ambient perspective' on Weyl geometries offered by
\textcite{Jia.Karydas.Leigh:2023}; the idea here is that consideration
of Weyl geometries in higher spacetime dimensions can be explanatory
of various aspects of Weyl geometries and conformal structures in four
spacetime dimensions.  Given this, one might wonder whether the
ambient perspective can shed light on Weyl's uniqueness result in four
dimensions.  In addition, one might seek to construct Galilean or
Carrollian Weyl ambient geometries, and then ask whether the
uniqueness (or, in the Carroll case, non-uniqueness) result of the
associated Weyl-type theorem can be illuminated from this
perspective.\footnote{We are grateful to Weizhen Jia, Luca Ciambelli,
  and an anonymous referee for this suggestion.}

There are also more foundational/conceptual payoffs of this work.  As
mentioned in the introduction, taking the lead from
\textcite{Weyl:1921}, \textcite{EPS:1972} (\acronym{EPS}) sought to
provide a \emph{constructive axiomatisation} (in the sense of
\textcite{Reichenbach:1969}; for more on constructive axiomatic
approaches to spacetime theories see
\textcite{Adlam.Linnemann.Read:2025}) of (the kinematical structures
of) the general theory of relativity, which is to say that they sought
to (i)~build up projective and (relativistic) conformal structures
from elementary, empirically-informed axioms, and then (ii)~show how
those structures, together with further `compatibility' assumptions,
yield the familiar Lorentzian geometries of general relativity.  In
this article, we have provided the Weyl-style `uniqueness' result as a
complement to the \acronym{EPS}-style `existence' result for Galilean
physics provided by \textcite[ch.\ 4]{Adlam.Linnemann.Read:2025};
moreover, the failure of the Carrollian `uniqueness' result might
raise some \emph{prima facie} concerns about the possibility of
Carroll constructive axiomatics.\footnote{Cf.\
  \cref{fn:Carr_constr_ax}.  There are other \emph{prima facie}
  worries about the construction in the Carroll case---for example,
  \textcite{EPS:1972} make substantial use of light-and-mirror
  constructions, but how to make sense of these when all signals are
  tachyonic, as is the case in Carroll spacetimes?}

Another foundational question in the vicinity of our work has to do
with the logic of experimental tests of spacetime theories.  Recently,
it has been argued by \textcite{Hansen.EtAl:2019} that many of the
classic experimental tests of general relativity would also be met by
modified Newtonian theories of gravity (with, generically, torsion
and/or non-metricity); \textcite{Wolf.Sanchioni.Read:2024} analyse
these claims and consider whether these experimental tests are in fact
best understood as testing projective/conformal structures.  Given
that we now have better control over e.g.\ conformal Carroll
structures, it would be worthwhile to consider the extent to which
these claims regarding experimental tests carry over to the Carroll
context.

And there is yet more work to be done besides the above.  Having now
thoroughly assessed the status of Weyl's \emph{theorem} in
non-relativistic contents, one might now turn to consideration of the
behaviour of the Weyl \emph{tensor} in such settings.  How does the
object behave in the relevant limits?  Does it remain an invariant of
the relevant conformal structure in each case?  Can it be used to
classify non-relativistic spacetimes, \emph{à la} the Petrov
classification of relativistic spacetimes?  Although some work has
already been undertaken in this vicinity (see in particular
\textcite{Dewar.Weatherall:2018, Dewar.Read:2020,
  Ehlers.Buchert:2009}), the results of this article should pave the
way to a more systematic treatment of the Galilei and Carroll
cases.\footnote{We undertake much of this work for the Galilean case
  in \textcite{Vigneron.Schwartz.Read:2026}.}

In the end, all this triangulates that there remains much work to be
done, along many axes, which draws upon the results of this article.
It is our hope that what we have achieved here lays the groundwork for
these future investigations.

\section*{Acknowledgements}

P.K.S.\ and Q.V.\ thank Pembroke College, Oxford, for hospitality
during a stay where the present research was conducted.  Q.V.\ was
supported in part by the European Research Council (\acronym{ERC})
under the European Union's Horizon 2020 research and innovation
program (grant agreement \acronym{ERC} advanced grant
740021-\acronym{ARTHUS}, principal investigator Thomas Buchert).
P.K.S.\ thanks Domenico Giulini for helpful discussions.  J.R.\ thanks
Niels Linnemann, Eleanor March, and the audiences of talks at
\acronym{UC} Irvine and at the Perimeter Institute (in the latter case
especially Luca Ciambelli and Weizhen Jia) for helpful discussions.

\appendix
\crefalias{section}{appendix} 

\section{The classification of connections on Carroll manifolds}
\label{sec:app_class_conn_Carroll}

Here, we will discuss the classification of connections on Carroll
manifolds in terms of freely specifiable tensor fields, based on
\textcite{Vigneron.Barzegar.Read:2025}.

\begin{constr}
  Let $(v, \gamma)$ be a Carroll structure on $M$.  In
  \textcite{Vigneron.Barzegar.Read:2025}, it was shown that any linear
  connection $\nabla$ on $M$ satisfies the two conditions
  \begin{subequations}
  \begin{align}
    \label{eq:Car_conn_cond_non-metr}
    v^\mu \nabla_\rho \gamma_{\mu\nu}
    &= - \gamma_{\mu\nu} \nabla_\rho v^\mu \\
    \shortintertext{and}
    \label{eq:Car_conn_cond_tors}
    T_{(\mu\nu)\rho} v^\rho
    &= \frac{1}{2} v^\rho \nabla_\rho \gamma_{\mu\nu}
      + \gamma_{\rho(\mu} \nabla_{\nu)} v^\rho
      - \frac{1}{2} (\mathcal L_v \gamma)_{\mu\nu}
  \end{align}
  \end{subequations}
  where $T$ is the torsion, whose first index has been lowered with
  $\gamma$.  Further it was shown that, fixing a unit timelike
  one-form $\tau$ (i.e.\ $\tau \in \Omega^1(M)$, $\tau(v) = 1$), the
  difference of any connection $\nabla$ to a reference connection
  $\accentset{v,\tau}{\nabla}$ (fully determined by $v, \gamma, \tau$)
  can be expressed algebraically in terms of
  \begin{enumerate}[label=(\arabic*),nosep]
  \item $v$, $\gamma$, $\tau$, and the contravariant space metric $h$
    with respect to $\tau$,
  \item the connection's non-metricities $\nabla \gamma$, $\nabla v$
    and torsion $T$ (subject to the above identities), as well as
  \item the field $\nabla_{(\mu} \tau_{\nu)}$,
  \end{enumerate}
  where the dependence on $\nabla \gamma$, $\nabla v$, $T$, and
  $\nabla_{(\mu} \tau_{\nu)}$ is affine.  (An explicit formula may be
  found in \textcite[Proposition 2.1]{Vigneron.Barzegar.Read:2025}.)

  However, \textcite{Vigneron.Barzegar.Read:2025} did not study to a
  full extent which components of these fields may be specified
  independently of each other, i.e.\ how to fully \emph{classify} a
  connection in terms of these fields.  We are now going to deduce
  this classification.  For this discussion, denote by $P$ the spatial
  projector along $\tau$.

  \begin{subequations}
  First considering the identity \eqref{eq:Car_conn_cond_non-metr}, we
  see that when arbitrarily fixing the $v$-non-metricity
  \begin{equation}
    \nabla_\rho v^\mu ,
  \end{equation}
  this determines the temporal part $v^\mu \nabla_\kappa
  \gamma_{\mu\nu}$ of the $\gamma$-non-metricity, leaving free
  precisely its purely spatial part (w.r.t.\ $\tau$) on its last two
  indices,
  \begin{equation}
    P^\rho_\mu P^\sigma_\nu \nabla_\kappa \gamma_{\rho\sigma} \; .
  \end{equation}
  With the non-metricities fixed, the identity
  \eqref{eq:Car_conn_cond_tors} fixes the part $T_{(\mu\nu)\rho}
  v^\rho$ of the torsion.  As discussed in the main text in
  \cref{constr:Car_torsion_free}, the remaining free components of the
  torsion then are those of the form
  \begin{equation}
    \left(T_{[\mu\nu]\rho} v^\rho \; , \;
    P^\rho_\mu P^\sigma_\nu T_{\kappa\rho\sigma} \; , \;
    (\delta^\rho_\sigma - P^\rho_\sigma) \tensor{T}{^\sigma_{\mu\nu}}
    \right).
  \end{equation}
  Finally, regarding the field $\nabla_{(\mu} \tau_{\nu)}$, one can
  show that given the non-metricities and torsion, it is fully
  determined by knowing only its purely spatial part\footnote{This may
    be seen as follows.  First, performing a spacelike--timelike
    decomposition, $\nabla_{(\mu} \tau_{\nu)}$ is determined by its
    purely spatial part $P^\rho_\mu P^\sigma_\nu \nabla_{(\rho}
    \tau_{\sigma)}$, its purely temporal part $v^\mu v^\nu \nabla_\mu
    \tau_\nu$, and its mixed part $P^\rho_\mu v^\nu \nabla_{(\rho}
    \tau_{\nu)} = P^\rho_\mu v^\nu \nabla_\rho \tau_\nu - P^\rho_\mu
    v^\nu \nabla_{[\rho} \tau_{\nu]}$.  Next, covariantly
    differentiating $\tau_\mu v^\mu = 1$, we see that the
    non-metricity $\nabla v$ determines the expression $v^\mu
    \nabla_\rho \tau_\mu$, and hence the purely temporal part and the
    first term of the mixed part.  Finally, the antisymmetric
    covariant derivative $\nabla_{[\rho} \tau_{\nu]}$ is determined by
    the exterior derivative $\D\tau$ and the torsion expression
    $\tau_\kappa \tensor{T}{^\kappa_{\rho\nu}}$, thus fixing the
    second term of the mixed part.}
  \begin{equation}
    \hat\chi_{\mu\nu}
    := P^\rho_\mu P^\sigma_\nu \nabla_{(\rho} \tau_{\sigma)} \; .
  \end{equation}
  \end{subequations}
  These fields may be understood as sections of the following vector
  bundles:
  \begin{subequations}
  the $v$-non-metricity is a section of
  \begin{equation}
    T^*M \otimes TM;
  \end{equation}
  the purely spatial part of the $\gamma$-non-metricity on its last
  two indices is a section of
  \begin{equation}
    \textstyle
    T^*M \otimes \bigvee^2 \ker v;
  \end{equation}
  the free torsion is a section of
  \begin{equation}
    \textstyle
    E_\text{free torsion} :=
    \left(\bigwedge^2 \ker v\right)
    \oplus
    \left(\ker v \otimes \bigwedge^2 \ker v\right)
    \oplus
    \left(\Span\{v\} \otimes \bigwedge^2 T^*M\right);
  \end{equation}
  and the purely spatial part $\hat\chi$ of $\nabla_{(\mu}
  \tau_{\nu)}$ is a section of
  \begin{equation}
    \textstyle
    \bigvee^2 \ker v.
  \end{equation}
  \end{subequations}
  Thus, the collection of all these freely specifiable fields
  \begin{subequations}
  \begin{equation}
    \label{eq:Car_conn_free_fields}
    \left(\nabla_\rho v^\mu \; , \;
    P^\rho_\mu P^\sigma_\nu \nabla_\kappa \gamma_{\rho\sigma} \; , \;
    T_{[\mu\nu]\rho} v^\rho \; , \;
    P^\rho_\mu P^\sigma_\nu T_{\kappa\rho\sigma} \; , \;
    (\delta^\rho_\sigma - P^\rho_\sigma) \tensor{T}{^\sigma_{\mu\nu}}
    \; , \;
    \hat\chi_{\mu\nu} \right)
  \end{equation}
  is a section of the vector bundle
  \begin{equation}
    \label{eq:Car_conn_free_fields_bundle}
    \textstyle
    E_\text{free fields} :=
    (T^*M \otimes TM)
    \oplus
    \left(T^*M \otimes \bigvee^2 \ker v\right)
    \oplus
    E_\text{free torsion}
    \oplus
    \bigvee^2 \ker v.
  \end{equation}
  \end{subequations}

  Now considering in addition the affine bundle $\mathrm{Conn}(TM)$ of
  linear connections on $M$, we have an affine bundle homomorphism
  \begin{equation}
    \Psi_\tau \colon \mathrm{Conn}(TM) \to E_\text{free fields} \; ,
  \end{equation}
  mapping a connection $\nabla$ to the field
  \eqref{eq:Car_conn_free_fields}.  Conversely, as shown by
  \textcite{Vigneron.Barzegar.Read:2025}, any connection $\nabla$ can
  be expressed affinely in terms of $\nabla \gamma$, $\nabla v$, $T$,
  $\nabla_{(\mu} \tau_{\nu)}$; by our above construction, these may
  further be expressed solely via \eqref{eq:Car_conn_free_fields}.
  This provides an affine bundle homomorphism
  \begin{equation}
    \Phi_\tau \colon E_\text{free fields} \to \mathrm{Conn}(TM)
  \end{equation}
  of which we know $\Phi_\tau \circ \Psi_\tau =
  \mathrm{id}_{\mathrm{Conn}(TM)}$.  We are now going to show that the
  ranks of $E_\text{free fields}$ and $\mathrm{Conn}(TM)$ are equal,
  such that for dimensional reasons $\Phi_\tau$ and $\Psi_\tau$ are
  mutually inverse isomorphisms.

  Writing $\dim M = n + 1$, the ranks of the direct summands of
  $E_\text{free fields}$ are
  \begin{subequations}
  \begin{align}
    \rk(T^*M \otimes TM)
    &= (n+1)^2 , \\
    \textstyle
    \rk\left(T^*M \otimes \bigvee^2 \ker v\right)
    &= (n+1) \cdot \frac{n(n+1)}{2}
      = \frac{n(n+1)^2}{2} \; , \displaybreak[0]\\
    \rk(E_\text{free torsion})
    &= \frac{n(n-1)}{2}
      + n \cdot \frac{n(n-1)}{2}
      + 1 \cdot \frac{(n+1)n}{2} \nonumber\\
    &= (1+n) \cdot \frac{n(n-1)}{2}
      + \frac{n(n+1)}{2} \nonumber\\
    &= \frac{n(n^2-1)}{2}
      + \frac{n(n+1)}{2} \nonumber\\
    &= \frac{n^2(n+1)}{2} \; , \\
    \textstyle
    \rk(\bigvee^2 \ker v)
    &= \frac{n(n+1)}{2} \; ,
  \end{align}
  \end{subequations}
  respectively.  Thus, we have
  \begin{align}
    \rk(E_\text{free fields})
    &= (n+1)^2
      + \frac{n(n+1)^2}{2}
      + \frac{n^2(n+1)}{2}
      + \frac{n(n+1)}{2} \nonumber\\
    &= (n+1)^2 + \frac{n(n+1)^2}{2} + \frac{n(n+1)^2}{2} \nonumber\\
    &= (n+1)^3 .
  \end{align}
  This is indeed equal to the rank of $\mathrm{Conn}(TM)$.  Therefore
  we may indeed conclude that $\Psi_\tau$ is an isomorphism of affine
  bundles.

  That is, choosing a unit timelike one-form $\tau$, \emph{any linear
    connection $\nabla$ on $M$ is uniquely determined by, and uniquely
    determines, the corresponding field
    \eqref{eq:Car_conn_free_fields}, which is a section of
    $E_\text{\emph{free fields}}$
    \eqref{eq:Car_conn_free_fields_bundle}}.
\end{constr}

\begin{rem}
  Note that regarding the independent degrees of freedom in the
  definition of a linear connection, it is crucial that with fixed
  non-metricities and torsion \emph{only the purely spatial part}
  $\hat\chi_{\mu\nu} = P^\rho_\mu P^\sigma_\nu \nabla_{(\rho}
  \tau_{\sigma)}$ of $\nabla_{(\mu} \tau_{\nu)}$ may be freely
  specified.  This was not explicitly realised in the discussion in
  \textcite{Vigneron.Barzegar.Read:2025}, where instead the whole
  quantity $\chi_{\mu\nu} = \nabla_{(\mu} \tau_{\nu)}$ was considered.

  The explicit counting of dimensions that we performed above shows
  that the field \eqref{eq:Car_conn_free_fields} captures all the
  freedom in the definition of a linear connection, expressed with
  respect to the Carroll structure, without any `double-counting'.
\end{rem}

\section{The free torsion of connections on conformal Carroll
  manifolds in conformally invariant form}
\label{sec:app_free_torsion_conf_Carroll}

Here, we will discuss how the freely specifiable torsion components of
connections on manifolds with a conformal Carroll structure may be
written in a conformally invariant way.

\begin{constr}
  \label{constr:Car_torsion_free_conf_inv}
  Let $[v, \gamma]$ be a conformal Carroll structure on $M$, and
  $[\tau]$ a timelike co-direction with respect to it.  Further, let
  $P$ be the spatial projector along $[\tau]$.

  In \cref{constr:Car_torsion_free}, we have seen that the components
  of the torsion that may be freely specified when defining a
  connection in terms of its non-metricities may be written in the
  form
  \begin{equation}
    \label{eq:Car_torsion_free_2}
    \left(T_{[\mu\nu]\rho} v^\rho \; , \;
    P^\rho_\mu P^\sigma_\nu T_{\kappa\rho\sigma} \; , \;
    (\delta^\rho_\sigma - P^\rho_\sigma) \tensor{T}{^\sigma_{\mu\nu}}
    \right)
  \end{equation}
  with respect to $[\tau]$ (cf.\ \eqref{eq:Car_torsion_free}), where a
  choice of representative $(v, \gamma)$ of the conformal structure
  was made.  In the first two terms, the first index on $T$ was
  lowered with $\gamma$.  We are now going to write the free torsion
  components in a form that is \emph{independent} of the choice of
  representative of the conformal structure, i.e.\ conformally
  invariant.

  For that, let $h$ be the contravariant space metric of $(v, \gamma)$
  with respect to $[\tau]$, defined by the requirements\footnote{Note
    that the second of these equations depends only on $[\tau]$.}
  \begin{equation}
    \label{eq:contr_space_metric}
    h^{\mu\nu} = h^{\nu\mu} , \quad
    h^{\mu\nu} \tau_\nu = 0 , \quad
    h^{\mu\nu} \gamma_{\nu\rho} = P^\mu_\rho \; .
  \end{equation}
  The conformally invariant form of the free torsion components is
  then given by
  \begin{equation}
    \label{eq:Car_torsion_free_conf_inv}
    \left(h^{\mu\kappa} \tensor{\gamma}{_{\lambda[\kappa}}
      \tensor{T}{^\lambda_{\nu]\sigma}}
      (\delta^\sigma_\rho - P^\sigma_\rho) \; , \;
    P^\lambda_\kappa P^\rho_\mu P^\sigma_\nu
      \tensor{T}{^\kappa_{\rho\sigma}} \; , \;
    (\delta^\rho_\sigma - P^\rho_\sigma) \tensor{T}{^\sigma_{\mu\nu}}
    \right).
  \end{equation}
  The second and third terms in \eqref{eq:Car_torsion_free_conf_inv}
  are manifestly conformally invariant, since the spatial projector
  $P$ along $[\tau]$ is.  Further, under a conformal rescaling
  $(v,\gamma) \mapsto (\e^{- \frac{1}{2} \lambda } v, \e^\lambda
  \gamma)$ the contravariant space metric transforms as $h \mapsto
  \e^{- \lambda} h$.  Hence, the first term in
  \eqref{eq:Car_torsion_free_conf_inv} is conformally invariant as
  well.  Of course, however, the free torsion components
  \eqref{eq:Car_torsion_free_conf_inv} still depend on the choice of
  timelike co-direction $[\tau]$.
\end{constr}

Next, we are going to discuss the space where the conformally
invariant form of the free torsion lives.  For this, we need a few
prerequisites:
\begin{constr}
  \label{constr:decomp_self-adj}
  \begin{enumerate}[label=(\roman*)]
  \item Let $V$ be a finite-dimensional real vector space and $g \in
    \bigvee^2 V^*$ a non-degenerate inner product on it (of any
    signature).

    For a linear map $X \colon V \to V$, the condition of being
    self-adjoint or anti-self-adjoint with respect to $g$,
    \begin{equation}
      \forall v,w \in V : g(v, X(w)) = \pm g(X(v), w),
    \end{equation}
    is invariant under rescaling of $g$: a linear map $X$ is
    (anti-)self-adjoint with respect to $g$ if and only if it is
    (anti-)self-adjoint with respect to $\e^\lambda g$, for any
    $\lambda \in \R$.  Therefore, \emph{the notion of
      (anti-)self-adjointness depends only on the conformal class}
    $[g]$ of the inner product.  In particular, given only a conformal
    class $[g]$ of non-degenerate inner products on $V$, the space
    \begin{equation}
      \so(V,[g])
      = \{X \colon V \to V \; \text{linear}
        \mid X \; \text{anti-self-adjoint w.r.t.} \; [g]\}
    \end{equation}
    is well-defined.\footnote{Similarly, the condition of a linear map
      $R \colon V \to V$ being orthogonal with respect to $g$,
      \begin{equation*}
        \forall v,w \in V : g(R(v), R(w)) = g(v, w),
      \end{equation*}
      is also invariant under rescaling of $g$, and hence only depends
      on the conformal class $[g]$.  Therefore, the orthogonal group
      $\mathrm{O}(V,[g])$ with respect to a conformal class of inner
      products is well-defined, and $\so(V,[g])$ is its Lie algebra.}

    Further, this implies that \emph{the decomposition of a linear
      endomorphism $X \colon V \to V$ into its self-adjoint and
      anti-self-adjoint parts with respect to a conformal class $[g]$
      is well-defined}.  Concretely, taking any non-degenerate inner
    product $g \in \bigvee^2 V^*$, with components $g_{ab}$, and
    denoting its inverse by $g^{-1} \in \bigvee^2 V$, with components
    $g^{ab}$, the decomposition of $X$ is given by
    \begin{subequations}
    \begin{align}
      \tensor{X}{^a_b}
      &= g^{ac} g_{dc} \tensor{X}{^d_b} \nonumber\\
      &= g^{ac} \left( \tensor{g}{_{d(c}} \tensor{X}{^d_{b)}}
        + \tensor{g}{_{d[c}} \tensor{X}{^d_{b]}} \right) \nonumber\\
      &= g^{ac} \tensor{g}{_{d(c}} \tensor{X}{^d_{b)}}
        + g^{ac} \tensor{g}{_{d[c}} \tensor{X}{^d_{b]}} \\
      \intertext{with the first summand being the self-adjoint and the
      second being the anti-self-adjoint part.  Equivalently, the
      decomposition can be written as}
      \tensor{X}{^a_b}
      &= g_{bd} \tensor{g}{^{c(d}_{\vphantom{c}}} \tensor{X}{^{a)}_c}
        + g_{bd} \tensor{g}{^{c[d}_{\vphantom{c}}} \tensor{X}{^{a]}_c}
        \; .
    \end{align}
    \end{subequations}
    Scaling $g$ with $\e^\lambda$, the inverse $g^{-1}$ is scaled with
    $\e^{-\lambda}$, such that the decomposition stays invariant.
  \item \label{constr:decomp_self-adj_ii}
    Let $M$ be a differentiable manifold of dimension $n + 1$, $n \ge
    1$, and $(v, \gamma)$ a Carroll structure on it.  $\gamma$ induces
    a natural positive-definite bundle metric $\gamman$ on the vector
    bundle $\ker v \subset T^*M$ of spacelike covectors.  Concretely,
    choosing any timelike co-direction $[\tau]$ and considering the
    contravariant space metric $h$ of $(v, \gamma)$ with respect to
    $[\tau]$ (defined by \eqref{eq:contr_space_metric}), $\gamman$ can
    be expressed as
    \begin{equation}
      \gamman(\alpha, \beta) = h^{\mu\nu} \alpha_\mu \beta_\nu
    \end{equation}
    for $\alpha, \beta \in \Gamma(\ker v)$ (i.e.\ $\gamman$ is the
    restriction of $h$ to $\ker v$).

    Conformally rescaling $(v,\gamma) \mapsto (\e^{- \frac{1}{2}
      \lambda } v, \e^\lambda \gamma)$, this bundle metric gets
    rescaled according to $\gamman \mapsto \e^{-\lambda} \cdot
    \gamman$.  In particular, a conformal Carroll structure
    $[v,\gamma]$ on $M$ induces a positive definite conformal bundle
    metric $[\gamman]$ on $\ker v$.  Thus, according to our previous
    discussion, the bundle
    \begin{equation}
      \so(\ker v, [\gamman]) \subset \mathrm{End}(\ker v)
    \end{equation}
    is well-defined, given a conformal Carroll structure
    $[v,\gamma]$.\footnote{Note that $\ker v$ is invariant under
      rescaling of $v$, and hence well-defined given only a conformal
      Carroll structure.}  Any vector bundle endomorphism of $\ker v$
    may then be decomposed into its self-adjoint and anti-self-adjoint
    parts with respect to $[\gamman]$, and the anti-self-adjoint part
    is a section of $\so(\ker v, [\gamman])$.

    Given a timelike co-direction $[\tau]$, we may identify vector
    bundle endomorphisms of $\ker v$ with vector bundle endomorphisms
    of $T^*M$ which are \emph{purely spatial with respect to
      $[\tau]$}, i.e.\ which take values in $\ker v$ and vanish on
    $\Span\{\tau\}$.\footnote{Concretely, the identification map reads
      \begin{equation*}
        \mathrm{End}(\ker v) \ni \hat X
        \xmapsto{\cong} \iota \circ \hat X \circ \tilde{P}
        \in \left\{X \in \mathrm{End}(T^*M) \; \middle| \;
          X|_{\Span\{\tau\}} = 0, \;
          \mathrm{im}(X) \subseteq \ker v\right\},
      \end{equation*}
      where $\iota \colon \ker v \to T^*M$ is the inclusion and
      $\tilde{P} \colon T^*M \to \ker v$ is the co-restriction of $P$,
      the spatial projector along $[\tau]$, to $\ker v$.  The inverse
      to this map is restriction and co-restriction of $X$ to $\ker
      v$.}  Under this identification, the decomposition of such a
    purely spatial map $X \colon T^*M \to T^*M$ into its self-adjoint
    and anti-self-adjoint parts with respect to $[\gamman]$ reads
    \begin{equation}
      \tensor{X}{_\mu^\nu}
      = h^{\nu\sigma}
          \tensor{\gamma}{_{\rho(\sigma}^{\vphantom{\rho}}}
          \tensor{X}{_{\mu)}^\rho}
        + h^{\nu\sigma}
          \tensor{\gamma}{_{\rho[\sigma}^{\vphantom{\rho}}}
          \tensor{X}{_{\mu]}^\rho} \; ,
   \end{equation}
   where $h^{\mu\nu}$ is the contravariant space metric with respect
   to $[\tau]$.

   Further, under transposition (i.e., abstractly speaking, taking
   dual maps), purely spatial (w.r.t.\ $[\tau]$) endomorphisms of
   $T^*M$ correspond to endomorphisms of $TM$ which are purely spatial
   w.r.t.\ $[\tau]$ in the sense of vanishing on $\Span\{v\}$ and
   taking values in $\ker\tau$.  Therefore, the decomposition of a
   purely spatial map $Y \colon TM \to TM$ according to
   \begin{equation}
     \label{eq:Car_decomp_self-adj}
      \tensor{Y}{^\mu_\nu}
      = h^{\mu\sigma}
          \tensor{\gamma}{_{\rho(\sigma}^{\vphantom{\rho}}}
          \tensor{Y}{^\rho_{\nu)}}
        + h^{\mu\sigma}
          \tensor{\gamma}{_{\rho[\sigma}^{\vphantom{\rho}}}
          \tensor{Y}{^\rho_{\nu]}}
    \end{equation}
    ---lowering the first index with $\gamma$, decomposing into
    symmetric and antisymmetric parts, and raising the first index
    again with $h$---corresponds to the decomposition of endomorphisms
    of $\ker v$ into self-adjoint and anti-self-adjoint parts with
    respect to $[\gamman]$, under this identification.\footnote{That
      is, under the identification of purely spatial endomorphisms of
      $TM$ first with purely spatial endomorphisms of $T^*M$ via
      transposition, and then with endomorphisms of $\ker v$ by
      restriction and co-restriction to $\ker v$.}
  \end{enumerate}
\end{constr}

\begin{constr}
  \label{constr:Car_torsion_free_space}
  Let $[v, \gamma]$ be a conformal Carroll structure on $M$, and
  $[\tau]$ a timelike co-direction with respect to it.  Further, let
  $P$ be the spatial projector along $[\tau]$.

  In \cref{constr:Car_torsion_free_conf_inv}, we have seen that when
  defining a connection on $M$ in terms of its non-metricities w.r.t.\
  $[v, \gamma]$, the freely specifiable torsion components w.r.t.\
  $[\tau]$ may be written in the conformally invariant form
  \eqref{eq:Car_torsion_free_conf_inv}
  \begin{equation}
    \label{eq:Car_torsion_free_conf_inv_2}
    \left(h^{\mu\kappa} \tensor{\gamma}{_{\lambda[\kappa}}
      \tensor{T}{^\lambda_{\nu]\sigma}}
      (\delta^\sigma_\rho - P^\sigma_\rho) \; , \;
    P^\lambda_\kappa P^\rho_\mu P^\sigma_\nu
      \tensor{T}{^\kappa_{\rho\sigma}} \; , \;
    (\delta^\rho_\sigma - P^\rho_\sigma) \tensor{T}{^\sigma_{\mu\nu}}
    \right).
  \end{equation}
  We will now discuss the space in which this form of the free torsion
  lives.

  The second term in \eqref{eq:Car_torsion_free_conf_inv_2}---the
  `purely spatial' projection of the torsion w.r.t.\ $[\tau]$---may be
  seen as a section of the vector bundle $\ker\tau \otimes \bigwedge^2
  \ker v$.  Similarly, the third term---the temporal projection
  w.r.t.\ $[\tau]$ on the first index---is a section of $\Span\{v\}
  \otimes \bigwedge^2 T^*M$.

  In the first term in \eqref{eq:Car_torsion_free_conf_inv_2}, the
  third index of the torsion is temporally projected.  Hence, due to
  the anti-symmetry of the torsion, the second index becomes purely
  spatial (i.e.\ contractions of the second index with elements of
  $\mathrm{span}\{v\}$ vanish).  The first index of the torsion is
  spatially projected by contraction with $\gamma$.  Thus, comparing
  to \eqref{eq:Car_decomp_self-adj}, we see that we may interpret the
  first term in \eqref{eq:Car_torsion_free_conf_inv_2} as follows: we
  temporally project the third index of the torsion, and on the first
  two indices, we spatially project and then take the
  anti-self-adjoint part w.r.t.\ the conformal bundle metric
  $[\gamman]$ on $\ker v$, in the sense of
  \cref{constr:decomp_self-adj}~\ref{constr:decomp_self-adj_ii}.
  Therefore, implicitly identifying purely spatial endomorphisms of
  $TM$ with endomorphisms of $\ker v$ as in
  \cref{constr:decomp_self-adj}~\ref{constr:decomp_self-adj_ii}, the
  first term in \eqref{eq:Car_torsion_free_conf_inv_2} is a section of
  $\so(\ker v, [\gamman]) \otimes \Span\{\tau\}$.\footnote{Note that
    the \emph{constrained} part of the torsion, which according to
    \eqref{eq:Car_torsion} is $T_{(\mu\nu)\rho} v^\rho$ or, in
    conformally invariant form, $h^{\mu\kappa}
    \tensor{\gamma}{_{\lambda(\kappa}}
    \tensor{T}{^\lambda_{\nu)\sigma}} (\delta^\sigma_\rho -
    P^\sigma_\rho)$, is the corresponding \emph{self-adjoint} part
    w.r.t.\ $[\gamman]$ of the spatial projection of the torsion on
    its first two indices, with the third index temporally projected.}

  Combined, we have thus seen that and how the conformally invariant
  form of the free torsion \eqref{eq:Car_torsion_free_conf_inv_2} of a
  connection on a conformal Carroll manifold with respect to a
  timelike co-direction $[\tau]$ is a section of the bundle
  \begin{equation}
    \label{eq:Car_torsion_free_space}
    \textstyle
    \Big(\so(\ker v, [\gamman]) \otimes \Span\{\tau\}\Big)
    \oplus
    \left(\ker\tau \otimes \bigwedge^2 \ker v\right)
    \oplus
    \left(\Span\{v\} \otimes \bigwedge^2 T^*M\right).
  \end{equation}

  Note that this value space \eqref{eq:Car_torsion_free_space} of the
  conformally invariant form of the free torsion depends on the choice
  of timelike co-direction $[\tau]$.  (This is different to the case
  of Galilei geometry (\cite{Schwartz:2025}), where the \emph{value}
  $P^\rho_\sigma \tensor{T}{^\sigma_{\mu\nu}}$ of the free torsion
  depends on the choice of timelike direction $[v]$, but the
  \emph{space} $\ker\tau \otimes \bigwedge^2 T^*M$ in which it takes
  values does not.)  We may get rid of this $[\tau]$ dependence of the
  value space, however at the cost of giving up conformal
  invariance---see the discussion in \cref{constr:Car_torsion_free}.
\end{constr}

\printbibliography

\end{document}
